\documentclass[11pt, reqno]{amsart}

\usepackage{amsmath}
\usepackage{amsfonts}
\usepackage{amssymb}
\usepackage{amsthm}
\usepackage[scr=dutchcal]{mathalpha}
\usepackage{mathtools}
\mathtoolsset{showonlyrefs}
\usepackage{dsfont}
\usepackage{tikz-cd}
\usepackage[italicdiff]{physics}
\usepackage{bm}
\usepackage{enumitem}
\usepackage{todonotes}
\usepackage{graphicx}
\usepackage{color}
\usepackage[colorlinks,linkcolor=blue,urlcolor=blue,citecolor=blue,destlabel=true]{hyperref}
\usepackage{comment}
\usepackage{cancel}
\numberwithin{equation}{section}

\makeatletter
\newtheorem*{rep@theorem}{\rep@title}
\newcommand{\newreptheorem}[2]{%
\newenvironment{rep#1}[1]{%
 \def\rep@title{#2 \ref{##1}}%
 \begin{rep@theorem}}%
 {\end{rep@theorem}}}
\makeatother

\newtheorem{thm}{Theorem}[section]
\newtheorem*{thm*}{Theorem}
\newtheorem{theorem}[thm]{Theorem}
\newreptheorem{theorem}{Theorem}

\newtheorem{prop}[thm]{Proposition}
\newtheorem{proposition}[thm]{Proposition}

\newtheorem{lemma}[thm]{Lemma}

\theoremstyle{definition}
\newtheorem{defn}[thm]{Definition}
\newtheorem*{defn*}{Definition}
\newtheorem{definition}[thm]{Definition}

\newtheorem{remark}[thm]{Remark}

\newtheorem{assumption}[thm]{Assumption}

\newcommand{\bbC}{\mathbb{C}}

\newcommand{\bbN}{\mathbb{N}}

\newcommand{\bbQ}{\mathbb{Q}}
\newcommand{\bbR}{\mathbb{R}}

\newcommand{\bbZ}{\mathbb{Z}}

\newcommand{\cB}{\mathfrak{B}}

\newcommand{\cE}{\mathcal{E}}
\newcommand{\cF}{\mathcal{F}}
\newcommand{\cG}{\mathcal{G}}

\newcommand{\sD}{\mathscr{D}}

\newcommand{\sH}{\mathscr{H}}

\newcommand{\sK}{\mathscr{K}}

\newcommand{\fB}{\mathfrak{B}}

\newcommand{\fM}{\mathfrak{M}}

\newcommand{\rr}{\mathbb{R}}

\newcommand{\defeq}{\vcentcolon=}

\DeclareMathOperator{\vecspan}{span}

\newcommand{\tn}[1]{\textnormal{#1}}

\DeclareMathOperator{\essran}{\tn{ess\,ran}}

\newcommand{\hilbH}{\sH}
\newcommand{\hilbK}{\sK}
\newcommand{\HS}{{\mathfrak{S}_2}}
\newcommand{\TC}{{\mathfrak{S}_1}}
\newcommand{\1}{\mathds{1}}
\newcommand{\B}{B}
\renewcommand{\ae}{\textnormal{a.e.\ }}
\newcommand{\inorm}[3]{\|#1\|_{#3}}
\newcommand{\supp}{\mathrm{supp}}

\usepackage{amsaddr}

\title[Foundations of Double Operator Integrals]{Foundations of Double Operator Integrals with a Variant Approach to the Nonseparable Case}

\thanks{\textsuperscript{A)}Also with the Center for Quantum Mathematics and Physics, University of California, Davis}
\thanks{\textsuperscript{B)}Also with the Department of Mathematics, University of California, Davis}
\thanks{\textsuperscript{C)}Also with the Simons Collaboration on Global Categorical Symmetries}

\begin{document}

\author[Ferydouni]{Robert Ferydouni\textsuperscript{\textnormal{1,a)}}}
\address{\textsuperscript{1)}Department of Mathematics, University of California, Davis, Davis, CA}
\author[Spiegel]{Daniel D.\ Spiegel\textsuperscript{2,a,b,c)}}
\address{\textsuperscript{2)}Department of Mathematics, Harvard University, Cambridge, MA}
\address{dspiegel@math.harvard.edu}

\begin{abstract}
We aim to give a self-contained and detailed yet simplified account of the foundations of the theory of double operator integrals, in order to provide an accessible entry point to the theory. We make two new contributions to these foundations: (1) a new proof of the existence of the product of two projection-valued measures, which allows for the definition of the double operator integral for Hilbert--Schmidt operators, and (2) a variant approach to the integral projective tensor product on arbitrary (not necessarily separable) Hilbert spaces using a somewhat more explicit norm than has previously been given. We prove the Daletskii--Krein formula for strongly differentiable perturbations of a densely defined self-adjoint operator and conclude by reviewing an application of the theory to quantum statistical mechanics.
\end{abstract}

\maketitle

\tableofcontents

\section{Introduction}

A definition of a double operator integral seeks to give precise meaning to an expression of the form
\begin{equation}\label{eq:doi}
	\iint \phi(x,y) \dd{E(x)} T \dd{F(y)},
\end{equation}
where $\phi$ is a scalar-valued function, $E$ and $F$ are projection-valued measures in Hilbert spaces $\sH$ and $\sK$, respectively, and $T$ is an operator from $\sK$ to $\sH$. Similarly, a multiple operator integral is a precise interpretation of an object of the form
\[
	\int \phi(x_1,\ldots, x_n) \dd{E_1(x_1)}T_1 \dd{E_2(x_2)}T_2 \cdots  \dd{E_{n-1}(x_{n-1})}T_{n-1} \dd{E_n(x_n)}.
\]

Double operator integrals were first considered by Daletskii and Krein \cite{DaletskiiKrein}, in a context where $E_s$ is the spectral measure of a family of bounded self-adjoint operators $A(s)$ that are continuously differentiable in the norm topology with respect to a single real variable $s$. A double operator integral is then used to make sense of the equation
\[
\frac{df(A(s))}{ds} = \iint \frac{f(x) - f(y)}{x - y} \dd{E_s(x)} A'(s) \dd{E_s(y)}
\]
for $C^2$ functions $f$. This equation is now known as the \emph{Daletskii--Krein formula}.
Double operator integrals were subsequently provided with a well-developed theory by Birman and Solomyak \cite{BirmanSolomyakI,BirmanM.Sh1968SDOI,MR348494}. By now, double and multiple operator integrals have proven to be useful tools for many problems in operator algebras and perturbation theory and various approaches to defining \eqref{eq:doi} have emerged (e.g., \cite{PagterSukochev,DoddsDoddsSukochevZaninII,Nikitopolous}). There are also a handful of survey articles \cite{Birman_Solomyak_DoubleOperatorIntegral,PellerSurvey} and textbooks \cite{SkripkaTomskova,HiaiKosaki} that cover their history, theory, and applications.

A rigorous approach to the double operator integral \eqref{eq:doi} pioneered by Birman and Solomyak  in \cite{BirmanSolomyakI} can be outlined as follows. Starting from projection-valued measures $E:\Sigma_X \to \fB(\sH)$ and $F:\Sigma_Y \to \fB(\sK)$  on measurable spaces $(X, \Sigma_X)$ and $(Y, \Sigma_Y)$, one constructs a projection-valued measure $\cG:\Sigma_{X \times Y} \to \fB(\HS(\sK, \sH))$, where $\Sigma_{X \times Y}$ is the $\sigma$-algebra generated by rectangles with sides in $\Sigma_X$ and $\Sigma_Y$, and $\HS(\sK, \sH)$ is the Hilbert space of Hilbert--Schmidt operators $\sK \to \sH$. The projection-valued measure $\cG$ can be characterized as the unique projection-valued measure on $\Sigma_{X \times Y}$ satisfying
\[
	\cG(\Gamma \times \Delta)T = E(\Gamma)TF(\Delta)
\]
for all $\Gamma \in \Sigma_X$, $\Delta \in \Sigma_Y$, and $T \in \HS(\sK, \sH)$. With $\cG$ in hand, one may define $\int \phi \dd{\cG}$ using the conventional theory of integration with projection-valued measures for bounded measurable functions $\phi:X \times Y \to \bbC$ (e.g., \cite{Murphy}). One identifies
\[
	\iint \phi(x,y) \dd{E(x)} T \dd{F(y)} \defeq \int \phi \dd{\cG} T
\]
for all $T \in \HS(\sK, \sH)$ and bounded measurable $\phi:X \times Y \to \bbC$. To extend the definition of the double operator integral to arbitrary $T \in \fB(\sK, \sH)$, one restricts to a special class of functions $\phi$ for which $\int \phi \dd{\cG}$ maps the set of trace-class operators $\TC(\sK, \sH)$ into itself. Defining $\widetilde \phi(y,x) = \phi(x,y)$ and a projection-valued measure $\widetilde \cG : \Sigma_{Y \times X} \to \fB(\HS(\sH, \sK))$ similarly to $\cG$, we obtain an operator $\int \widetilde \phi \dd{\widetilde \cG}$ that maps $\TC(\sH, \sK)$ into itself. Taking the Banach space adjoint of $\int \widetilde \phi \dd{\widetilde \cG}$ and applying the isometric isomorphism
\[
	\fB(\sK, \sH) \to \TC(\sH, \sK)^*, \quad T \mapsto (S \mapsto \Tr(ST))
\]
gives a definition of $\int \phi \dd{\cG} T$ for $T \in \fB(\sK, \sH)$.

For most of its history the theory of double operator integrals has been confined to the case where $\sH$ and $\sK$ are separable Hilbert spaces. Only recently was the general case provided a full treatment by Nikitopoulos \cite{Nikitopolous}, who follows the method of Birman and Solomyak above. One of the difficulties with the nonseparable case lies in adequately describing a space of functions $\phi$ for which $\int \phi \dd{\cG}$ maps $\TC(\sK, \sH)$ into itself. In the separable case, it was proven by Peller \cite{PellerSchurMultipliers} that this condition is equivalent to $\phi$ admitting a decomposition
\begin{equation}\label{eq:decomposition}
	\phi(x,y) = \int_Z \alpha(x,z)\beta(y,z) \dd{\nu(z)} \quad \tn{for $\cG$-a.e.\ $(x,y) \in X \times Y$}
\end{equation}
for some finite positive measure space $(Z, \Sigma_Z, \nu)$ and bounded measurable functions $\alpha:X \times Z \to \bbC$ and $\beta:Y \times Z \to \bbC$. Whether these conditions are equivalent in the nonseparable case is unknown, but one can show that the existence of a decomposition \eqref{eq:decomposition} implies that $\int \phi \dd{\cG}$ maps $\TC(\sK, \sH)$ into itself \cite[Thm.~4.2.9]{Nikitopolous}.

One may therefore consider the space of functions $\phi$ satisfying \eqref{eq:decomposition} (or some generalization thereof, e.g., \ \cite[Def.~4.1.2]{Nikitopolous}). This is a vector space that one wishes to endow with a complete norm. In the separable case, Peller gives the norm
\begin{equation}\label{eq:peller_norm}
	\norm{\phi} = \inf \int_Z \norm{\alpha(\,\cdot\,, z)}_{L^\infty(E)} \norm{\beta(\,\cdot \, , z)}_{L^\infty(F)} \dd{\nu(z)},
\end{equation}
where the infimum is taken over all representations of $\phi$ as in \eqref{eq:decomposition}. In the nonseparable case, one must contend with the issue that $z \mapsto \norm{\alpha(\, \cdot \, , z)}_{L^\infty(E)}$ and $z \mapsto \norm{\beta(\, \cdot \, , z)}_{L^\infty(F)}$ need not be measurable. Nikitopoulos addresses this issue by replacing the integral in \eqref{eq:peller_norm} with an upper integral.\footnote{We refer to Nikitopoulos \cite[Rem.~4.1.3]{Nikitopolous} for further discussion of the subtleties surrounding the measurability of these functions.} By definition, if $(X, \Sigma_X, \mu)$ is a positive measure space and $f:X \to [0,\infty]$ is a (not necessarily measurable) function, then its \emph{upper integral} is
\[
	\overline{\int_X} f \dd{\mu} \defeq \inf\left\{ \int_X \tilde f \dd{\mu} \bigg| \tn{ $\tilde f:X \to [0,\infty]$ is measurable and $f \leq \tilde f$ $\mu$-a.e.} \right\}.
\]
One may also define a \emph{lower integral} similarly:
\[
	\underline{\int_X} f \dd{\mu} \defeq \sup\left\{ \int_X \tilde f \dd{\mu} \bigg| \tn{ $\tilde f:X \to [0,\infty]$ is measurable and $\tilde f \leq f$ $\mu$-a.e.} \right\}.
\]

The main new result of this article is that one also obtains a complete norm essentially by taking a lower integral, except that the supremum is taken over a natural and explicitly identified set of functions $\tilde f$. Precisely, we define $\fM_0$ to be the set of bounded measurable functions $\phi:X \times Y \to \bbC$ for which there exists a decomposition as in \eqref{eq:decomposition}. Then, we define
\[
	\norm{\phi}_{\fM_0} \defeq \inf_{(Z, \Sigma_Z, \nu, \alpha, \beta)} \sup_{v, w} \int_Z \norm{\alpha(\, \cdot \, ,  z)}_{L^\infty(E_{v,v})}\norm{\beta(\, \cdot \, , z)}_{L^\infty(F_{w,w})} \dd{\nu(z)},
\]
where the infimum is taken over all decompositions of $\phi$ as in \eqref{eq:decomposition}, the supremum is taken over all $v \in \sH$ and $w \in \sK$, and $E_{v,v}$ and $F_{w,w}$ are the positive measures $E_{v,v}(\,\cdot\,) = \ev{v, E(\,\cdot\,)v}$ and $F_{w,w}(\,\cdot\,) = \ev{w, F(\,\cdot\,)w}$. One can show that the integrand above is indeed measurable in all cases (see Proposition \ref{prop:esssup_measurable}). \S\ref{sec:integral_projective_tensor_products} is largely devoted to the proof of the following theorem.

\begin{reptheorem}{thm:integral_projective_tensor_product}
The set $\fM_0$ is a unital $*$-subalgebra of the $C^*$-algebra $B(X \times Y)$ of bounded measurable functions $X \times Y \to \bbC$. The function $\norm{\cdot}_{\fM_0}$ is a submultiplicative seminorm on $\fM_0$ invariant under complex conjugation. Defining $\pi:B(X \times Y) \to L^\infty(\cG)$ to be the canonical projection, the seminorm $\norm{\cdot}_{\fM_0}$ further satisfies
\[
	\norm{\pi(\phi)}_{L^\infty(\cG)} \leq \norm{\phi}_{\fM_0}
\]
for every $\phi \in \fM_0$. Defining $\fM = \pi(\fM_0)$, the seminorm $\norm{\cdot}_{\fM_0}$ descends to a well-defined norm on $\fM$ making $\fM$ into a unital Banach $*$-algebra.
\end{reptheorem}

	The difference between our approach and that of Nikitopoulos \cite{Nikitopolous} is somewhat subtle, so let us take a moment to compare the two in greater detail. Nikitopoulos defines a space of functions consisting of those $\phi \in L^\infty(\cG)$ that satisfy \eqref{eq:decomposition} for some $\sigma$-finite positive measure space $(Z, \Sigma_Z, \nu)$ and measurable functions $\alpha:X \times Z \to \bbC$ and $\beta:Y \times Z \to \bbC$ such that $\alpha( \cdot , z ) \in L^\infty(E)$ and $\beta( \cdot , z ) \in L^\infty(F)$ for all $z \in Z$ and such that
	\[
		\overline{\int_Z} \norm{\alpha(\, \cdot\,  , z)}_{L^\infty(E)} \norm{\beta(\, \cdot\, , z)}_{L^\infty(F)} \dd{\nu(z)} < \infty.
	\]
	Thus, every element of $\fM$ as we have defined it satisfies the criteria of Nikitopoulos above. When $\sH$ and $\sK$ are separable, the converse is also true. Whether this converse is generally true when $\sH$ and $\sK$ are nonseparable is unknown to the authors.

	Nikitopoulos equips his space of functions with what appears to be a different norm than we have defined. Namely, Nikitopoulos defines the norm
	\[
		\norm{\phi} = \inf_{\nu, \alpha, \beta} \overline{\int_Z} \norm{\alpha(\, \cdot\, , z)}_{L^\infty(E)} \norm{\beta(\, \cdot\, , z)}_{L^\infty(F)} \dd{\nu(z)},
	\]
	where the infimum is taken over all $\nu$, $\alpha$, and $\beta$ satisfying the criteria at the top of the previous paragraph. Again the authors are unaware of a $\phi \in \fM$ on which the two norms differ. Nonetheless, given a positive measure space $(Z, \Sigma_Z, \nu)$ and measurable $\alpha:X \times Z \to \bbC$ and $\beta :Y \times Z \to \bbC$, one in general has
	\begin{align*}
		&\sup_{v, w} \int_Z \norm{\alpha(\, \cdot \, ,  z)}_{L^\infty(E_{v,v})}\norm{\beta(\, \cdot \, , z)}_{L^\infty(F_{w,w})} \dd{\nu(z)} \\
		&\hspace{1in} \leq \overline{\int_Z} \norm{\alpha(\, \cdot\, , z)}_{L^\infty(E)} \norm{\beta(\, \cdot\, , z)}_{L^\infty(F)} \dd{\nu(z)}.
	\end{align*}
	Again when the Hilbert spaces are separable, equality holds. However, strict inequality can hold when the Hilbert spaces are nonseparable. An example is provided by choosing $\alpha$ and $\beta$ to both be the function given in the last example of \cite[Rem.~4.1.3]{Nikitopolous}.

Thus, Theorem \ref{thm:integral_projective_tensor_product} does not appear to follow from the results of \cite{Nikitopolous} and requires a new method of proof. In addition to proving this theorem, another main purpose of this article is to provide an accessible exposition of the foundations of double operator integrals in the context of arbitrary Hilbert spaces that is nonetheless detailed and self-contained. For simplicity and readability we have restricted our attention to double operator integrals, rather than multiple operator integrals, and consider only the ideals of Hilbert--Schmidt, trace-class, and bounded operators. While the exposition follows a similar trajectory to \cite{Nikitopolous,nikitopoulos2023higher} and uses the method of Birman and Solomyak \cite{BirmanSolomyakI}, some of our proofs have been simplified. 

In Theorem \ref{thm:double_PVM} we provide a new construction of the projection-valued measure $\cG:\Sigma_{X \times Y} \to \fB(\HS(\sK, \sH))$. 
The original construction of Birman and Solomyak constructs $\cG$ by a process similar to the proof of the Carath\'eodory extension theorem \cite{Birman_Solomjak_SpectralTheory}. Instead, we will show how the Carath\'eodory extension theorem can be combined with a representation theorem for bounded quadratic forms on a Hilbert space (Theorem \ref{thm:Hilbert_quadratic_form}) to construct $\cG$ more efficiently. Although this representation theorem is elementary and originally appeared in \cite{Lukes}, to the authors' knowledge this is the first time in the literature that it has appeared in English with a detailed proof.

 In Theorem \ref{thm:TC_to_TC} we have provided a streamlined proof that $\phi \in \fM$ implies $\int \phi \dd{\cG}$ maps $\TC(\sK, \sH)$ to $\TC(\sK, \sH)$ and through Proposition \ref{prop:ultraweakly_meas}, Theorem \ref{thm:doi_TC_mel}, and Theorem \ref{thm:doi_bounded_T} we have streamlined the derivation of the equation 
 \[
 	\Tr(S^* \int \phi \dd{\cG} T) = \int_Z \Tr(S^*\int_X \alpha_z \dd{E} T \int_Y \beta_z \dd{F}) \dd{\nu(z)},
 \]
 where $S \in \TC(\sK, \sH)$, $T \in \fB(\sK, \sH)$, and $\phi$ is decomposed as in \eqref{eq:decomposition}.  In Proposition \ref{prop:SOT_sequential_continuity} we provide a proof that $\int \phi \dd{\cG}:\fB(\sK, \sH) \to \fB(\sK, \sH)$ is sequentially continuous with respect to the strong operator topology on domain and codomain, which to the authors' knowledge has so far only been proved in the case of separable Hilbert spaces \cite{AzamovCareyDoddsSukochev}.

In \S\ref{sec:DK} we prove the Daletskii--Krein formula in the simple case where $f \in W_1(\bbR)$, where $W_1(\bbR)$ is the first Wiener space, defined in Definition \ref{def:Wiener_space}. We do note that Nikitopoulos \cite{nikitopoulos2023higher} has generalized to the nonseparable case the Daletskii--Krein formula and its higher-derivative counterparts for a broader class of functions, originally considered by Peller \cite{Pellerhigheroperatorderivatives}; see also the survey article \cite{AleksandrovPeller}. On the other hand, we will consider a different type of perturbation than has previously been considered. The works \cite{Pellerhigheroperatorderivatives,AleksandrovPeller} consider the Daletskii--Krein formula for functions of the form $A(s) = A_0 + sB$, where $A_0$ is a densely defined self-adjoint operator and $B$ is a bounded self-adjoint operator, giving cases where $f(A(s))$ is strongly differentiable \cite[Thm.~3.5.6]{AleksandrovPeller} and Fr\'echet differentiable, \cite[Thm.~3.5.11]{AleksandrovPeller} and \cite[Thm.~5.6]{Pellerhigheroperatorderivatives}. There also exist results in the literature on Fr\'echet differentiability of the function $B \mapsto f(A+B) - f(A)$ in certain cases, e.g., \cite[Thm.~1.2.2]{nikitopoulos2023higher} and \cite[Thm.~5.7]{AzamovCareyDoddsSukochev}. In contrast, we prove the strong differentiability, with derivative given by the Daletskii--Krein formula, for functions $s \mapsto f(A_0 + \Phi(s))$, where $A_0$ is a densely defined self-adjoint operator on a Hilbert space $\sH$ and $\Phi:J\to \fB(\sH)_\tn{sa}$ is a strongly differentiable function on an interval $J \subset \bbR$. This type of perturbation appears in quantum statistical mechanics \cite{NSY_quasilocality1}, for example.

Indeed, in \S\ref{sec:application} we review an application of this theory to quantum statistical mechanics that inspired the authors' interest in the subject. The results of \S\ref{sec:application} originally appeared in \cite{BMNS_automorphic_equivalence}, were expanded upon in \cite{NSY_quasilocality1}, and were inspired by work in the physics literature \cite{HastingsWen}. 
The setup is a parameter-dependent path of Hamiltonians $H(s)$, like $A(s)$ above, that have a piece of their spectrum isolated from the rest. With $P(s)$ as the projection onto the isolated part of the spectrum of $H(s)$, we will use the theory of double operator integrals to show the operator equation
\begin{equation}\label{eq:projection_evolution}
\frac{d}{ds} P(s) = i[D(s), P(s)]\,,
\end{equation}
after we show each side is a linear combination of double operator integrals applied to $\Phi'(s)$. Here, $D(s)$ is a self-adjoint operator, called the Hastings generator, which has an explicit expression in terms of $H(s)$ given in \S\ref{sec:application}. The main consequence of \eqref{eq:projection_evolution} is the existence of a parameter-dependent family of unitaries $U(s)$ on $\sH$ satisfying $P(s) = U(s) P(0) U(s)^*$. The unitaries arise as the unique strong solution to the linear differential equation
\[\frac{d}{ds} U(s) = iD(s) U(s), \quad U(0) = \mathds{1}\,.\]
We note that Kato \cite{Kato} showed the existence of unitaries $V(s)$ satisfying $P(s) = V(s) P(0) V(s)^*$, that arise as the unique solution to the linear differential equation
\[\frac{d}{ds} V(s) = [P'(s), P(s)] V(s), \quad V(0) = \mathds{1}\,.\]
A benefit of the theorem in \S\ref{sec:application} compared to Kato's theorem is the explicitness of the Hastings generator $D(s)$ and consequently $U(s)$, which \cite{BMNS_automorphic_equivalence} and \cite{NSY_quasilocality1} needed for their applications to quantum spin systems.

\subsection*{Notation, Terminology, and Conventions}
All our Hilbert spaces will be complex. Given Hilbert spaces $\sH$ and $\sK$, we let $\cB(\sK, \sH)$ denote the set of bounded linear maps $\sK \to \sH$. Given any set $\Omega$, any function $f:\Omega \rightarrow \cB(\sK, \sH)$, and any vectors $v \in \sH$ and $w \in \sK$, we let $f_{v,w}:\Omega \rightarrow \bbC$ be the function 
\[
f_{v,w}(\omega) \defeq \ev{v, f(\omega)w}.
\]
We let $\HS(\sK, \sH)$ denote the space of Hilbert--Schmidt operators $\sK \to \sH$ and $\TC(\sK, \sH)$ denote the space of trace-class operators $\sK \to \sH$.

The identity operator on a vector space is denoted by $\1$.

We adhere to certain mathematical physics conventions described below. We denote complex conjugation with an asterisk. All our inner products are linear in the second variable. Given Hilbert spaces $\sH$ and $\sK$ and vectors $v \in \sH$ and $w \in \sK$, we let $\ketbra{v}{w} : \sK \rightarrow \sH$ denote the operator
\[
\ketbra{v}{w}u \defeq \ev{w,u}v.
\]

Given measurable spaces $(X, \Sigma_X)$ and $(Y, \Sigma_Y)$, we let $\Sigma_X \times \Sigma_Y$ be the set of all measurable rectangles $A \times B$, where $A \in \Sigma_X$ and $B \in \Sigma_Y$. We let $\Sigma_{X \times Y}$ be the $\sigma$-algebra generated by $\Sigma_X \times \Sigma_Y$. We will always tacitly equip $X \times Y$ with the $\sigma$-algebra $\Sigma_{X \times Y}$. 

If $(X, \Sigma_X)$ is a measurable space, we let $\B(X)$ denote the Banach space of bounded measurable functions $X \rightarrow \bbC$, where $\bbC$ is equipped with its Borel $\sigma$-algebra, and $\B(X)$ is equipped with the supremum norm. We will more often work with the Banach spaces $L^\infty(E)$ and $L^\infty(E_{v,v})$, where $E$ is a projection-valued measure on a Hilbert space $\sH$ and $v \in \sH$. In fact we will frequently be in situations where there are multiple vectors $v$ to be considered at once. To make the notation less burdensome, we will abbreviate
\[
	\norm{\cdot}_{E} \defeq \norm{\cdot}_{L^\infty(E)} \qqtext{and} \norm{\cdot}_v \defeq \norm{\cdot}_{L^\infty(E_{v,v})},
\]
where these norms are the essential suprema of an equivalence class of functions with respect to the indicated measure. Alternatively we may think of these essential suprema as seminorms on $\B(X)$ and we will denote these seminorms with the same notation as in the display above. 

Finally, the characteristic function of a set $A \subset X$ is denoted $\chi_A$.

\subsection*{Acknowledgments}
The authors thank Bruno Nachtergaele for many helpful conversations and D.\ D.\ Spiegel also thanks Lorenzo Riva for helpful conversations. D.\ D.\ Spiegel thanks Jan Vyb\'iral for pointing out the reference \cite{Lukes} for Theorem \ref{thm:Hilbert_quadratic_form}. While performing this work, R.\ Ferydouni was supported by the National Science Foundation under Grant No.\ DMS 2510824. D.\ D.\ Spiegel was supported by the National Science Foundation under Grant No.\ DMS 2303063 and by the Simons Foundation Award 888992 as part of the Simons Collaboration on Global Categorical Symmetries.

\section{Constructing a PVM on Hilbert--Schmidt Operators}\label{sec:PVM}

Let $\sH$ and $\sK$ be complex Hilbert spaces, let $(X, \Sigma_X)$ and $(Y, \Sigma_Y)$ be nonempty measurable spaces, and let $E:\Sigma_X \rightarrow \cB(\sH)$ and $F:\Sigma_Y \rightarrow \cB(\sK)$ be projection-valued measures. In this section, we follow 
\cite[\S3.1]{Birman_Solomyak_DoubleOperatorIntegral}
and \cite{Birman_Solomyak_tensor_product_PVM} and construct a projection-valued measure $\cG:\Sigma_{X \times Y} \to \cB(\HS(\sK, \sH))$ from $E$ and $F$. The construction first defines $\cG$ on $\Sigma_X \times \Sigma_Y$ and then extends to $\Sigma_{X \times Y}$. We present a new and streamlined proof of this extension by using Theorem \ref{thm:Hilbert_quadratic_form}, a representation theorem for bounded quadratic forms on a Hilbert space.

\begin{proposition}[{\cite[\S3.1]{Birman_Solomyak_DoubleOperatorIntegral}}]
The functions 
\begin{alignat}{2}
\cE &: \Sigma_X \to \cB(\HS(\sK, \sH)),& \quad \cE(\Gamma)T &= E(\Gamma)T\\
\cF &: \Sigma_Y \to \cB(\HS(\sK, \sH)),& \quad \cF(\Delta)T &= TF(\Delta)
\end{alignat}
are projection-valued measures. Furthermore,
\[
	[\cE(\Gamma), \cF(\Delta)] = 0
\]
for all $\Gamma \in \Sigma_X$ and $\Delta \in \Sigma_Y$.
\end{proposition}

\begin{proof}
It is an easy exercise to see that $\cE$ and $\cF$ are well-defined projection-valued functions, that $\cE(X) = \cF(Y) = \1$, and that $[\cE(\Gamma), \cF(\Delta)] = 0$ for all $\Gamma \in \Sigma_X$ and $\Delta \in \Sigma_Y$. 

We observe that if $\Gamma, \Gamma' \in \Sigma_X$ are disjoint, then $\cE(\Gamma)$ and $\cE(\Gamma')$ are orthogonal projections since $E(\Gamma)$ and $E(\Gamma')$ are orthogonal projections. 
Now let $(\Gamma_k)_{k \in \bbN}$ be a sequence of pairwise disjoint elements of $\Sigma_X$ and let $\Gamma = \bigcup_{k \in \bbN} \Gamma_k$. We know $\sum_{k=1}^\infty \cE(\Gamma_k)$ converges to a projection in the strong operator topology on $\cB(\HS(\sK, \sH))$ since the partial sums are an increasing sequence of projections. Given $T \in \HS(\sK, \sH)$, $v \in \sH$, and $w \in \sK$, we have
\begin{align*}
\ev{v, \qty(\qty(\sum_{k=1}^\infty \cE(\Gamma_k))T)w}  &= \ev{\ketbra{v}{w}, \sum_{k=1}^\infty \qty(\cE(\Gamma_k)T)}_\HS\\
&= \sum_{k=1}^\infty \ev{\ketbra{v}{w}, E(\Gamma_k)T}_\HS\\
&= \sum_{k=1}^\infty \ev{v,E(\Gamma_k)Tw}\\
&= \ev{v, E(\Gamma)Tw}\\
&= \ev{v, \qty(\cE(\Gamma)T)w}.
\end{align*}
This proves that $\cE(\Gamma)T = \qty(\sum_{k=1}^\infty \cE(\Gamma_k))T$. Since $T$ was arbitrary, this proves that $\cE$ is countably additive. A similar calculation shows that $\cF$ is countably additive.
\end{proof}

\begin{proposition}[{\cite[\S3.1]{Birman_Solomyak_DoubleOperatorIntegral}}]\label{prop:curly_E_F}
If $\alpha \in \B(X)$, then
\[
\qty(\int \alpha \dd{\cE})(T) = \int \alpha \dd{E} \cdot T 
\]
and if $\beta \in \B(Y)$, then
\[
\qty(\int \beta \dd{\cF})(T) = T \cdot \int \beta \dd{F}.
\]
\end{proposition}

\begin{proof}
Let $x \in \sH$ and $y \in \sK$. Note that $\cE_{\ketbra{x}{y}, T} = E_{x, Ty}$. Then
\begin{align*}
\ev{x, \qty(\int \alpha \dd{\cE})(T)y} &= \ev{\ketbra{x}{y}, \qty(\int \alpha \dd{\cE})T}_\HS\\
&= \int \alpha \dd{\cE_{\ketbra{x}{y}, T}}\\
&= \int \alpha \dd{E_{x,Ty}}\\
&= \ev{x, \int \alpha \dd{E} Ty}.
\end{align*}
This proves the first identity. A similar calculation using $\cF_{\ketbra{x}{y}, T} = F_{T^*x,y}$ proves the second identity.
\end{proof}

For the next lemma, note that if $\lambda$ and $\mu$ are two complex measures on $\Sigma_X$ and $\Sigma_Y$, then there exists a unique complex measure $\lambda \times \mu$ on $\Sigma_{X \times Y}$ such that 
\[
(\lambda \times \mu)(\Gamma \times \Delta) = \lambda(\Gamma)\mu(\Delta)
\]
for all $\Gamma \in \Sigma_X$ and $\Delta \in \Sigma_Y$. The existence follows by breaking $\lambda$ and $\mu$ into linear combinations of four finite positive measures, extending the products of these finite positive measures to $\Sigma_{X \times Y}$, and assembling these product measures in the appropriate linear combination. For uniqueness, if $\nu$ and $\nu'$ are two complex measures on $\Sigma_{X \times Y}$ satisfying $\nu(\Gamma \times \Delta) = \nu'(\Gamma \times \Delta)$ for all $\Gamma \in \Sigma_X$ and $\Delta \in \Sigma_Y$, then the subset of $\Sigma_{X \times Y}$ where $\nu$ and $\nu'$ agree is a $\lambda$-system containing the $\pi$-system of measurable rectangles and is therefore all of $\Sigma_{X \times Y}$ by the Dynkin $\pi$--$\lambda$ theorem.

\begin{lemma}[{\cite[Thm.~1]{Birman_Solomyak_tensor_product_PVM}}]\label{lem:G_sigma_add_on_rectangles}
The function
\[
	\cG:\Sigma_X \times \Sigma_Y \to \cB(\HS(\sK, \sH)), \quad \cG(\Gamma \times \Delta) = \cE(\Gamma)\cF(\Delta)
\]
is projection-valued and $\sigma$-additive on $\Sigma_X \times \Sigma_Y$ with respect to the weak operator topology.
\end{lemma}

\begin{proof}
Each $\cG(\Gamma \times \Delta)$ is a projection since $\cE(\Gamma)$ and $\cF(\Delta)$ are commuting projections. 

We show $\cG$ is $\sigma$-additive. Let $\Gamma \in \Sigma_X$ and $\Delta \in \Sigma_Y$ and let $w,x \in \sH$ and $y,z \in \sK$. Then a quick calculation shows that
\[
\cG_{\ketbra{w}{z}, \ketbra{x}{y}}(\Gamma \times \Delta) = (E_{w,x} \times F_{y, z})(\Gamma \times \Delta).
\]
Since $E_{w,x} \times F_{y,z}$ is a complex measure on $\Sigma_{X \times Y}$, it is $\sigma$-additive on $\Sigma_X \times \Sigma_Y$. Thus, $\cG_{\ketbra{w}{z}, \ketbra{x}{y}}$ is $\sigma$-additive on $\Sigma_X \times \Sigma_Y$. By sesquilinearity, $\cG_{S, T}$ is $\sigma$-additive for any finite-rank operators $S, T \in \HS(\sK, \sH)$.

Let $(\Gamma_n \times \Delta_n)_{n =1}^\infty$ be a disjoint countable collection of measurable rectangles such that $\bigcup_{n =1}^\infty \Gamma_n \times \Delta_n \in \Sigma_X \times \Sigma_Y$. For distinct $m,n \in \bbN$, either $\Gamma_m \cap \Gamma_n = \varnothing$ or $\Delta_m \cap \Delta_n = \varnothing$. It follows that the projections $\cG(\Gamma_n \times \Delta_n)$ are pairwise orthogonal. In particular, we know that $\sum_{n =1}^\infty \cG(\Gamma_n \times \Delta_n)$ converges in the weak operator topology. Thus, for any finite-rank operators $S, T \in \HS(\sK, \sH)$,
\begin{align*}
\ev{S, \qty(\sum_{n=1}^\infty \cG(\Gamma_n \times \Delta_n))T} &= \sum_{n=1}^\infty \cG_{S,T}(\Gamma_n \times \Delta_n)\\
&= \cG_{S,T}\qty(\bigcup_{n =1}^\infty \Gamma_n \times \Delta_n) \\
&= \ev{S, \cG\qty(\bigcup_{n=1}^\infty \Gamma_n \times \Delta_n)T}.
\end{align*}
By density of the finite-rank operators in $\HS(\sK, \sH)$, we conclude that 
\[
	\sum_{n=1}^\infty \cG\qty(\Gamma_n \times \Delta_n) = \cG\qty(\bigcup_{n=1}^\infty \Gamma_n \times \Delta_n),
\]
so $\cG$ is $\sigma$-additive.
\end{proof}

We would now like to show that $\cG$ has a unique extension to a projection-valued measure on $\Sigma_{X \times Y}$. We will use the following theorem of Hilbert space theory. The only published version of this result that the authors are aware of is \cite[9.24~Lemma]{Lukes}. As we are not aware of this result being published in the English literature anywhere, we have given a full proof. The proof we give here first appeared in a similar form on MathOverflow \cite{MathMathMathOverflow} and includes many details not present in \cite{Lukes}.

\begin{theorem}\label{thm:Hilbert_quadratic_form}
Let $\hilbH$ be a Hilbert space and let $\rho:\sH \rightarrow \bbC$ be a function such that 
\begin{enumerate}
	\item[\tn{(1)}] $\rho(\lambda u) = \abs{\lambda}^2 \rho(u)$ for all $\lambda \in \bbC$ and $u \in \hilbH$,
	\item[\tn{(2)}] $\rho(u + v) + \rho(u - v) = 2\rho(u) + 2\rho(v)$,
	\item[\tn{(3)}] there exists $C \geq 0$ such that $\abs{\rho(u)} \leq C \norm{u}^2$ for all $u \in \hilbH$.
\end{enumerate}
Then the map $\sigma:\sH \times \sH \rightarrow \bbC$ defined by
\[
\sigma(u,v) = \frac{1}{4}\qty[\rho(u + v) - \rho(u - v) - i \rho(u + iv) + i \rho(u - iv)]
\]
is a sesquilinear form with $\norm{\sigma} \leq 2C$ and $\sigma(u,u) = \rho(u)$ for all $u \in \sH$. In particular, there exists a unique $A \in \cB(\hilbH)$ such that 
\[
\rho(u) = \ev{u, Au}
\]
for all $u \in \hilbH$. Furthermore, $\norm{A} \leq 2C$. If $\rho(u) \in \bbR$ for all $u \in \hilbH$, then $A$ is self-adjoint and $\norm{A} \leq C$. If $\rho(u) \geq 0$ for all $u \in \hilbH$, then $A$ is positive and $\norm{A} \leq C$.
\end{theorem}

\begin{proof}
Observe that $\abs{\rho(0)} \leq C\norm{0}^2 = 0$, so $\rho(0) = 0$. Then it is easily verified from the definition of $\sigma$ that $\sigma(u,u) = \rho(u)$ for all $u \in \sH$.

We note that a routine (but slightly messy) calculation gives
\begin{equation}\label{eq:sigma_biadditive}
\begin{aligned}
\sigma(u, v + w) &= \sigma(u, v) + \sigma(u,w)\\
\sigma(v+w, u) &= \sigma(v,u) + \sigma(w,u)
\end{aligned}
\end{equation}
for all $u,v , w \in \sH$.

Now we show that $\sigma$ is sesquilinear with respect to scalar multiplication. We first note that 
\begin{equation}\label{eq:sigma_zero}
	\sigma(u, 0) = \sigma(0, u) = 0
\end{equation}
for all $u \in \sH$ by definition of $\sigma$. It then follows easily by induction and \eqref{eq:sigma_biadditive} that $\sigma(u, nv) = n\sigma(u,v)$ and $\sigma(nu,v) = n\sigma(u,v)$ for all $n \in \qty{0} \cup \bbN$ and $u, v \in \sH$.  It also follows from \eqref{eq:sigma_biadditive} and \eqref{eq:sigma_zero} that $\sigma(u, -v) = -\sigma(u,v)$ and $\sigma(-u, v) = -\sigma(u,v)$. We see that $\sigma$ is bilinear with respect to the integers. Next, if $m,n \in \bbZ$ and $n \neq 0$, then
\begin{align*}
n\sigma\qty(u, \frac{m}{n}v) &= \sigma(u, mv) = m\sigma(u,v),\\
n\sigma\qty(\frac{m}{n}u, v) &= \sigma(mu,v) = m\sigma(u,v)
\end{align*}
proving that $\sigma$ is bilinear with respect to the rationals. Next, one can verify directly from the definition of $\sigma$ that
\begin{align*}
\sigma(u,iv) &= i\sigma(u,v)\\
\sigma(iu,v) &= -i\sigma(u,v).
\end{align*}

To prove that $\sigma$ is bilinear over the reals, we must first prove that $\sigma$ is bounded. Observe that
\begin{align*}
\abs{\sigma(u,v)} &\leq \frac{C}{4}\qty[\norm{u + v}^2 + \norm{u - v}^2 + \norm{u + iv}^2 + \norm{u - iv}^2] \\
&= C(\norm{u}^2 + \norm{v}^2).
\end{align*}
For any strictly positive  $\lambda \in \bbQ$, we have
\begin{align*}
\abs{\sigma(u, v)} = \abs{\sigma(\lambda u, \lambda^{-1}v)} \leq C\qty(\lambda^2 \norm{u}^2 + \lambda^{-2}\norm{v}^2).
\end{align*}
If $\norm{u} = 0$ or $\norm{v} = 0$, then $\sigma(u,v) = 0$, hence $\abs{\sigma(u,v)} \leq C \norm{u}\norm{v}$. If $\norm{u} \neq 0$ and $\norm{v} \neq 0$, then taking a sequence of strictly positive rationals $(\lambda_n)$ such that $\lambda_n \rightarrow \sqrt{\norm{v}/\norm{u}}$ yields
\[
\abs{\sigma(u,v)} \leq 2C\norm{u}\norm{v}.
\]
It follows easily from this that $\sigma$ is continuous with respect to the product topology on $\hilbH \times \hilbH$. 

Now for any real number $\lambda$, we may take a sequence of rationals $(\lambda_n)$ such that $\lambda_n \rightarrow \lambda$, whence continuity of $\sigma$ implies
\begin{align*}
\sigma(u, \lambda v) &= \lim_n \sigma(u, \lambda_n v) = \lim_n \lambda_n \sigma(u,v) = \lambda \sigma(u,v)\\
\sigma(\lambda u, v) &= \lim_n \sigma(\lambda_n u, v) = \lim_n \lambda_n \sigma(u,v) = \lambda \sigma(u,v).
\end{align*}

We have proven that $\sigma$ is a bounded sesquilinear form with $\norm{\sigma} \leq 2C$. Thus, there exists a unique operator $A \in \cB(\hilbH)$ such that $\norm{A} \leq 2C$ and $\sigma(u,v) = \ev{u, Av}$ for all $u, v \in \hilbH$. Suppose $\rho(u) \in \bbR$ for all $u \in \hilbH$. Then it is easily checked from the definition that $\sigma(u,v)^* = \sigma(v,u)$,
so $\sigma$ is Hermitian and therefore $A$ is self-adjoint. Furthermore,
\begin{align*}
\norm{A} = \sup_{\norm{u} \leq 1} \abs{\ev{u, Au}} = \sup_{\norm{u} \leq 1} \abs{\rho(u)} \leq C.
\end{align*}
If $\rho(u) \geq 0$ for all $u \in \hilbH$, then $\ev{u,Au} \geq 0$ for all $u \in \hilbH$, so $A$ is positive. The fact that $\norm{A} \leq C$ follows as above.
\end{proof}

\begin{theorem}[{cf.~\cite[Thms.~5.2.3 \& 5.2.4]{Birman_Solomjak_SpectralTheory}}]\label{thm:double_PVM}
There exists a unique projection-valued measure $\cG:\Sigma_{X \times Y} \to \cB(\HS(\sK, \sH))$ such that
\[
	\cG(\Gamma \times \Delta) = \cE(\Gamma)\cF(\Delta)
\]
for all $\Gamma \in \Sigma_X$ and $\Delta \in \Sigma_Y$.
\end{theorem}

\begin{proof}
For each $T \in \HS(\sK, \sH)$, we have a $\sigma$-additive function
\[
	\mu_T :\Sigma_X \times \Sigma_Y \to [0, \norm{T}^2_\HS], \quad \mu_{T}(\Lambda) = \ev{T, \cG(\Lambda)T}.
\]
By the Carath\'eodory extension theorem, each $\mu_T$ has a unique extension to a finite positive measure on $\Sigma_{X \times Y}$, which we shall also denote as $\mu_T$.

Given $\Lambda \in \Sigma_{X \times Y}$, define
\[
	\rho_\Lambda : \HS(\sK, \sH) \to [0,\infty), \quad \rho_\Lambda(T) = \mu_T(\Lambda).
\]
Note that
\[
\rho_\Lambda(T) \leq \mu_T(X \times Y) = \norm{T}^2_\HS.
\]
For any $\Lambda \in \Sigma_X \times \Sigma_Y$  and $T \in \HS(\hilbK, \hilbH)$, it follows easily from the definition that
\begin{align*}
\mu_{\lambda T}(\Lambda) &= \abs{\lambda}^2 \mu_T(\Lambda).
\end{align*}
The measures $\mu_{\lambda T}$ and $\abs{\lambda}^2 \mu_T$ are therefore equal on $\Sigma_{X \times Y}$ by the uniqueness clause of the Carath\'eodory extension theorem. Similarly, for any $\Lambda \in \Sigma_X \times \Sigma_Y$ and $S, T \in \HS(\sK, \sH)$, we have
\[
\mu_{S+T}(\Lambda) + \mu_{S-T}(\Lambda) = 2\mu_S(\Lambda) + 2\mu_T(\Lambda).
\]
Therefore the measures $\mu_{S+T} + \mu_{S - T}$ and $2\mu_{S} + 2\mu_{T}$ are again equal on $\Sigma_{X \times Y}$ by the uniqueness clause of the Carath\'eodory extension theorem. 

Thus, for each $\Lambda \in \Sigma_{X \times Y}$ we see that $\rho_\Lambda$ satisfies the hypotheses of Theorem \ref{thm:Hilbert_quadratic_form}, hence there exists a unique positive $\cG(\Lambda) \in \cB(\HS(\sK, \sH))$ such that $\norm{\cG(\Lambda)} \leq 1$ and $\ev{T, \cG(\Lambda)T} = \rho_\Lambda(T)$ for all $T \in \HS(\hilbK, \hilbH)$. By uniqueness, the function $\cG:\Sigma_{X \times Y} \rightarrow \cB(\HS(\hilbK, \hilbH))$ is an extension of the original $\cG:\Sigma_X \times \Sigma_Y \rightarrow \cB(\HS(\sK, \sH))$.

We must show that $\cG:\Sigma_{X \times Y} \rightarrow \cB(\HS(\hilbK, \hilbH))$ is a projection-valued measure. Let $(\Lambda_n)_{n \in \bbN}$ be a countable collection of disjoint elements of $\Sigma_{X \times Y}$. Then for every $T \in \HS(\sK, \sH)$, we have
\begin{align*}
	\ev{T, \cG\qty(\bigcup_{n=1}^\infty \Lambda_n)T} &= \rho_{\bigcup_{n=1}^\infty \Lambda_n}(T) = \mu_T\qty(\bigcup_{n=1}^\infty \Lambda_n) \\
	&= \sum_{n=1}^\infty \mu_T(\Lambda_n) = \sum_{n=1}^\infty \rho_{\Lambda_n}(T) = \sum_{n=1}^\infty \ev{T, \cG(\Lambda_n)T}.
\end{align*}
Since this occurs for every $T \in \HS(\sK, \sH)$, we see that the series $\sum_{n=1}^\infty \cG(\Lambda_n)$ converges in the weak operator topology to $\cG(\bigcup_{n=1}^\infty \Lambda_n)$.

It remains to show that $\cG$ is projection-valued. Since we already know that $\cG(\Lambda)$ is positive for all $\Lambda \in \Sigma_{X \times Y}$, the fact that it is a projection will follow immediately once we prove the identity
\begin{equation}\label{eq:G_multiply_identity}
	\cG(\Lambda_0)\cG(\Lambda_1) = \cG(\Lambda_0 \cap \Lambda_1)
\end{equation}
for all $\Lambda_0, \Lambda_1 \in \Sigma_{X \times Y}$. We can easily verify \eqref{eq:G_multiply_identity} for all $\Lambda_0, \Lambda_1 \in \Sigma_X \times \Sigma_Y$ by using the definition of $\cG$ and the fact that $\cE$ and $\cF$ are projection-valued measures with commuting ranges. 

Now fix $\Lambda_0 \in \Sigma_{X} \times \Sigma_Y$. Suppose $\Lambda_1 \in \Sigma_{X \times Y}$ satisfies \eqref{eq:G_multiply_identity}. By additivity of $\cG$ we know $\cG(\Lambda_1^c) = \1 - \cG(\Lambda_1)$, hence
\[
	\cG(\Lambda_0)\cG(\Lambda_1^c) = \cG(\Lambda_0) - \cG(\Lambda_0 \cap \Lambda_1) = \cG(\Lambda_0 \cap \Lambda_1^c),
\]
with the second equality again following by additivity of $\cG$. This verifies \eqref{eq:G_multiply_identity} for $\Lambda_1^c$. If $(\Lambda_n)_{n \in \bbN}$ is a countable collection of disjoint elements of $\Sigma_{X \times Y}$ for which the identity \eqref{eq:G_multiply_identity} holds for each $\Lambda_n$ in place of $\Lambda_1$, then by $\sigma$-additivity of $\cG$ we have
\begin{align*}
\cG(\Lambda_0)\cG\qty(\bigcup_{n=1}^\infty \Lambda_n) &= \cG(\Lambda_0)\qty(\sum_{n=1}^\infty \cG(\Lambda_n)) \\
&= \sum_{n=1}^\infty \cG(\Lambda_0)\cG(\Lambda_n) \\
&= \sum_{n=1}^\infty \cG(\Lambda_0 \cap \Lambda_n)\\
&= \cG\qty(\Lambda_0 \cap \bigcup_{n=1}^\infty \Lambda_n).
\end{align*}
This verifies \eqref{eq:G_multiply_identity} for $\bigcup_{n=1}^\infty \Lambda_n$. Thus, the set of all elements $\Lambda_1 \in \Sigma_{X \times Y}$ for which \eqref{eq:G_multiply_identity} holds is a $\lambda$-system containing the $\pi$-system $\Sigma_X \times \Sigma_Y$. By the Dynkin $\pi$--$\lambda$ theorem we can conclude that \eqref{eq:G_multiply_identity} holds for all $\Lambda_0 \in \Sigma_X \times \Sigma_Y$ and $\Lambda_1 \in \Sigma_{X \times Y}$. Now fixing $\Lambda_1 \in \Sigma_{X \times Y}$ and applying an essentially identical argument, we see that the Dynkin $\pi$--$\lambda$ theorem implies \eqref{eq:G_multiply_identity} holds for all $\Lambda_0 \in \Sigma_{X \times Y}$. We have now proven that $\cG$ is a projection-valued measure.

Suppose $\tilde \cG$ is any other projection-valued measure fulfilling the requirements of the theorem. The set of elements of $\Sigma_{X \times Y}$ on which $\tilde \cG$ and $\cG$ agree is a $\lambda$-system containing the $\pi$-system $\Sigma_X \times \Sigma_Y$. Thus, the Dynkin $\pi$--$\lambda$ theorem implies that $\tilde \cG$ and $\cG$ agree on all of $\Sigma_{X \times Y}$.
\end{proof}

Let us collect some basic facts about $\cG$.

\begin{proposition}\label{prop:G_measure_properties}
Let $\cG$ be the projection-valued measure constructed in Theorem \ref{thm:double_PVM}.
\begin{enumerate}[label=\tn{(\alph*)}]
	\item For every $v_1,v_2 \in \sH$ and $w_1, w_2 \in \sK$, the measure $\cG_{\ketbra{v_1}{w_1},\ketbra{v_2}{w_2}}$ can be written as a linear combination of measures of the form $\cG_{\ketbra{v}{w},\ketbra{v}{w}}$ where $v \in \vecspan\{v_1,v_2\}$ and $w \in \vecspan\{w_1,w_2\}$.
	\item For every $S, T \in \HS(\sK, \sH)$, the complex measure $\cG_{S,T}$ is in the closed linear span $\qty{\cG_{\ketbra{v}{w}, \ketbra{v}{w}}: v \in \sH, w \in \sK}$.
	\item For every $v_1, v_2 \in \sH$ and $w_1, w_2 \in \sK$, we have 
	\[
		\cG_{\ketbra{v_1}{w_1}, \ketbra{v_2}{w_2}} = E_{v_1, v_2} \times F_{w_2, w_1}.
	\]
\end{enumerate}
\end{proposition}

\begin{proof}
(a) Let $v_1,v_2 \in \sH$ and $w\in \sK$. We compute:
\begin{align*}
\cG_{\ketbra{v_1+v_2}{w},\ketbra{v_1+v_2}{w}}   &= \cG_{\ketbra{v_1}{w}, \ketbra{v_1}{w}} + \cG_{\ketbra{v_2}{w}, \ketbra{v_2}{w}}\\
&\qquad + \cG_{\ketbra{v_1}{w},\ketbra{v_2}{w}} + \cG_{\ketbra{v_2}{w},\ketbra{v_1}{w}}.
\end{align*}
Likewise,
\begin{align*}
\cG_{\ketbra{v_1+iv_2}{w},\ketbra{v_1+iv_2}{w}} &= \cG_{\ketbra{v_1}{w},\ketbra{v_1}{w}} + \cG_{\ketbra{v_2}{w}, \ketbra{v_2}{w}}\\
&\qquad + i\cG_{\ketbra{v_1}{w},\ketbra{v_2}{w}} - i\cG_{\ketbra{v_2}{w},\ketbra{v_1}{w}}.
\end{align*}
We see that
\begin{align*}
\cG_{\ketbra{v_1}{w},\ketbra{v_2}{w}} &= \frac{1}{2}\left[\cG_{\ketbra{v_1+v_2}{w},\ketbra{v_1+v_2}{w}} - \cG_{\ketbra{v_1}{w},\ketbra{v_1}{w}} - \cG_{\ketbra{v_2}{w},\ketbra{v_2}{w}} \right.\\
&\qquad \left. -i\cG_{\ketbra{v_1+iv_2}{w},\ketbra{v_1+iv_2}{w}} +i \cG_{\ketbra{v_1}{w},\ketbra{v_1}{w}} + i \cG_{\ketbra{v_2}{w},\ketbra{v_2}{w}}\right].
\end{align*}
This proves the theorem in the case where $w_1 = w_2$. Now let $v_1,v_2 \in \sH$ and $w_1,w_2 \in \sK$. We compute:
\begin{align*}
\cG_{\ketbra{v_1}{w_1+w_2},\ketbra{v_2}{w_1+w_2}}  &= \cG_{\ketbra{v_1}{w_1},\ketbra{v_2}{w_1}} + \cG_{\ketbra{v_1}{w_2},\ketbra{v_2}{w_2}}\\
&\qquad + \cG_{\ketbra{v_1}{w_1},\ketbra{v_2}{w_2}} + \cG_{\ketbra{v_1}{w_2},\ketbra{v_2}{w_1}}.
\end{align*}
Likewise,
\begin{align*}
\cG_{\ketbra{v_1}{w_1+iw_2},\ketbra{v_2}{w_1+iw_2}} &= \cG_{\ketbra{v_1}{w_1},\ketbra{v_2}{w_1}} + \cG_{\ketbra{v_1}{w_2}, \ketbra{v_2}{w_2}}\\
&\qquad - i\cG_{\ketbra{v_1}{w_1},\ketbra{v_2}{w_2}} + i\cG_{\ketbra{v_1}{w_2},\ketbra{v_2}{w_1}}.
\end{align*}
Combining these gives
\begin{align*}
\cG_{\ketbra{v_1}{w_1},\ketbra{v_2}{w_2}} &= \frac{1}{2}\left[\cG_{\ketbra{v_1}{w_1+w_2},\ketbra{v_2}{w_1+w_2}} - \cG_{\ketbra{v_1}{w_1},\ketbra{v_2}{w_1}} - \cG_{\ketbra{v_1}{w_2},\ketbra{v_2}{w_2}} \right.\\
& \left. +i\cG_{\ketbra{v_1}{w_1+iw_2},\ketbra{v_2}{w_1+iw_2}} -i \cG_{\ketbra{v_1}{w_1},\ketbra{v_2}{w_1}} - i \cG_{\ketbra{v_1}{w_2},\ketbra{v_2}{w_2}}\right]
\end{align*}
whereby the proof now concludes by reference to the case where $w_1 = w_2$.

(b) If $S$ and $T$ are finite-rank, then part (a) implies that 
\[
	\cG_{S,T} \in \vecspan\{\cG_{\ketbra{v}{w},\ketbra{v}{w}}:v \in \sH, w \in \sK\}.
\]
The result now follows since the map $(S,T) \mapsto \cG_{S,T}$ is continuous.

(c) It is an easy exercise to show that $\cG_{\ketbra{v_1}{w_1}, \ketbra{v_2}{w_2}}$ and $E_{v_1, v_2} \times F_{w_2, w_1}$ agree on all measurable rectangles $\Gamma \times \Delta \in \Sigma_X \times \Sigma_Y$. The set of elements on which these two complex measures agree is then a $\lambda$-system containing the $\pi$-system $\Sigma_X \times \Sigma_Y$, and is therefore all of $\Sigma_{X \times Y}$ by the Dynkin $\pi$--$\lambda$ theorem.
\end{proof}

\begin{proposition}[{\cite[\S3.1]{Birman_Solomyak_DoubleOperatorIntegral}}]\label{prop:totally_separated_variable_integral}
Let $\cG$ be the projection-valued measure constructed in Theorem \ref{thm:double_PVM}. If $\alpha:X \to \bbC$ and $\beta:Y \to \bbC$ are bounded and measurable and $\phi:X \times Y \to \bbC$ is defined by $\phi(x,y) = \alpha(x)\beta(y)$, then
\[
	\int_{X \times Y} \phi \dd{\cG} = \int_X \alpha \dd{\cE} \int_Y \beta \dd{\cF} = \int_Y \beta \dd{\cF} \int_X \alpha \dd{\cE}.
\]
\end{proposition}

\begin{proof}
For any $v_1, v_2 \in \sH$ and $w_1, w_2 \in \sK$, Proposition \ref{prop:G_measure_properties} yields
\begin{align*}
\int_{X \times Y} \phi \dd{\cG_{\ketbra{v_1}{w_1},\ketbra{v_2}{w_2}}} &= \int_{X \times Y} \phi \dd{(E_{v_1,v_2} \times F_{w_2,w_1})}\\
&= \int_{X} \alpha \dd{E_{v_1,v_2}} \int_Y \beta \dd{F_{w_2,w_1}}\\
&= \ev{v_1, \int_X \alpha \dd{E} \ketbra{v_2}{w_2} \int_Y \beta \dd{F} w_1}\\
&= \ev{\ketbra{v_1}{w_1}, \int_X \alpha \dd{\cE} \int_Y \beta \dd{\cF} \ketbra{v_2}{w_2}}\\
&= \ev{\ketbra{v_1}{w_1}, \int_Y \beta \dd{\cF} \int_X \alpha \dd{\cE} \ketbra{v_2}{w_2}}.
\end{align*}
Since $\HS(\sK, \sH)$ is the closed linear span of the rank-one operators, the result follows.
\end{proof}

We now study the sets of measure zero and relations of absolute continuity that arise  among the measures $\cG_{S,T}$. We begin with some considerations for arbitrary projection-valued measures.

\begin{definition}
Given a projection-valued measure $E:\Sigma_X \to \cB(\sH)$ and $v \in \sH$, we define
\[
\sH_v = \overline{\vecspan\{E(\Delta)v:\Delta \in \Sigma_X\}}.
\]
We say $v_1,v_2 \in \sH$ are \emph{spectrally orthogonal} if $\sH_{v_1} \perp \sH_{v_2}$, or equivalently if
\[
	\ev{E(\Delta_1)v_1, E(\Delta_2)v_2} = 0
\]
for all $\Delta_1,\Delta_2 \in \Sigma_X$. Note that if $v_1, v_2 \in \sH$ are spectrally orthogonal, then they are orthogonal, as seen by setting $\Delta_1 = \Delta_2 = X$. 
\end{definition}

The following theorem generalizes  \cite[Thm.~7.3.4]{Birman_Solomjak_SpectralTheory} to the case of a nonseparable Hilbert space $\sH$.

\begin{theorem}\label{thm:sep_spectral_dominator}
Let $E:\Sigma_X \to \cB(\sH)$ be a projection-valued measure. Let $\sH_0$ be a separable linear subspace of $\sH$. There exists $v \in \sH$ such that $E(\Delta)v = 0$ implies $E(\Delta)u = 0$ for all $u \in \sH_0$. In particular,
\[
	E_{u_1,u_2} \ll E_{v,v}
\]
whenever $u_1, u_2 \in \sH$ and either $u_1 \in \sH_0$ or $u_2 \in \sH_0$.
\end{theorem}

\begin{proof}
Let $(e_n)_{n \in \bbN}$ be a collection of orthogonal vectors in $\sH_0$ whose span is dense in $\sH_0$. Define $v_1 = e_1$. Then define $v_2$ to be the projection of $e_2$ onto the orthogonal complement of $\sH_{v_1}$ (relative to $\sH$). We see that $v_1$ and $v_2$ are spectrally orthogonal and $e_1, e_2 \in \sH_{v_1} \oplus \sH_{v_2}$. In general, once $v_1,\ldots, v_n$ have been defined, we define $v_{n+1}$ to be the projection of $e_{n+1}$ onto the orthogonal complement of $\bigoplus_{k=1}^n \sH_{v_k}$. Then $v_{n+1}$ is spectrally orthogonal to all $v_k$ with $k \leq n$ and $e_{n+1} \in \bigoplus_{k=1}^{n+1} \sH_{v_k}$. This defines a collection of spectrally orthogonal vectors $(v_n)_{n \in \bbN}$ in $\sH$ such that $\sH_0 \subset \bigoplus_{k=1}^\infty \sH_{v_k}$.

Now, let
\[
	v = \sum_{k=1}^\infty k^{-1}v_k.
\]
Since the vectors $(v_n)_{n \in \bbN}$ are orthogonal and the sequence $(k^{-1})_{k \in \bbN}$ is square summable, the series above converges and $v$ is a well-defined element of $\sH$. Suppose $\Delta \in \Sigma_X$ such that $E(\Delta)v = 0$.  Since the vectors $(v_n)_{n \in \bbN}$ are spectrally orthogonal, this implies that $E(\Delta)v_k = 0$ for all $k \in \bbN$. It follows from the definition of $\sH_{v_k}$ that $E(\Delta)u = 0$ for all $u \in \sH_{v_k}$. Since this holds for all $k$, we have $E(\Delta)u = 0$ for all $u \in \bigoplus_{k=1}^\infty \sH_{v_k}$, in particular for all $u \in \sH_0$.
\end{proof}

We will now discuss absolute continuity for the specific projection-valued measure $\cG$, but first we remind the reader of the following basic result from measure theory. The lemma below is similar to Exercise 4 in Section 5.2 of \cite{Cohn_Measure_Theory}. We have provided a proof for the convenience of the reader.

\begin{lemma}\label{lem:product_meas_absolute_continuity}
Let $(X_1, \Sigma_1)$ and $(X_2, \Sigma_2)$ be measurable spaces. Suppose $\mu_1, \mu_2$ are $\sigma$-finite positive measures on $X_1$ and $X_2$, respectively, and $\lambda_1$ and $\lambda_2$ are complex measures on $X_1$ and $X_2$, respectively. If $\lambda_1 \ll \mu_1$ and $\lambda_2 \ll \mu_2$, then $\lambda_1 \times \lambda_2 \ll \mu_1 \times \mu_2$.
\end{lemma}

\begin{proof}
We know from \cite[Lem.~11, \S III.11]{DunfordSchwartz} that $\abs{\lambda_1 \times \lambda_2} = \abs{\lambda_1} \times \abs{\lambda_2}$, and $\lambda_1 \times \lambda_2 \ll \abs{\lambda_1 \times \lambda_2}$, so it suffices to prove that $\abs{\lambda_1} \times \abs{\lambda_2} \ll \mu_1 \times \mu_2$. Furthermore, the fact that $\lambda_i \ll \mu_i$ implies that $\abs{\lambda_i} \ll \mu_i$. 

Let $h_i \in L^1(\mu_i)$ be the Radon--Nikodym derivative of $\abs{\lambda_i}$ with respect to $\mu_i$ and note that we can take $h_i$ to be nonnegative. Suppose $E \subset X_1 \times X_2$ is measurable and $(\mu_1 \times \mu_2)(E) = 0$. Then Tonelli's theorem gives
\begin{align*}
(\abs{\lambda_1} \times \abs{\lambda_2})(E) &= \int_{X_1 \times X_2} \chi_E \dd{(\abs{\lambda_1} \times \abs{\lambda_2})}\\
&= \int_{X_1} \qty(\int_{X_2} \chi_E(x,y) h_2(y) \dd{\mu_2(y)}) \dd{\abs{\lambda_1}(x)} \\
&= \int_{X_1}\qty(\int_{X_2} \chi_E(x,y)h_1(x)h_2(y) \dd{\mu_2(y)})\dd{\mu_1(x)}\\
&= \int_{E} h_1 h_2 \dd{\mu_1 \times \mu_2} = 0
\end{align*}
since $(\mu_1 \times \mu_2)(E) = 0$. This proves that $\abs{\lambda_1} \times \abs{\lambda_2} \ll \mu_1 \times \mu_2$.
\end{proof}

Finally, we give a result that will be important in the following section.

\begin{theorem}\label{thm:G_TT_absolute_cont_product}
Let $\cG$ be the projection-valued measure from  Theorem \ref{thm:double_PVM}. For all $S,T \in \HS(\sK, \sH)$, there exist $v \in \sH$ and $w \in \sK$ such that
\[
	\cG_{S,T} \ll  E_{v,v} \times F_{w,w}.
\]
\end{theorem}

\begin{proof}
Note that $\cG_{S,T} \ll \cG_{T,T}$, so it suffices to prove the theorem assuming $S = T$. Let $\sH_0 = T(\sK)$ and let $\sK_0 = (\ker T)^\perp$. Since $T$ is Hilbert--Schmidt, we know both $\sH_0$ and $\sK_0$ are separable. Therefore by Theorem \ref{thm:sep_spectral_dominator} there exists $v \in \sH$ such that $E(\Gamma)v = 0$ implies $E(\Gamma)u = 0$ for all $u \in \sH_0$ and there exists $w \in \sK$ such that $F(\Delta)w = 0$ implies $F(\Delta)u = 0$ for all $u \in \sK_0$.

Consider a finite-rank operator $T_0 = \sum_{k=1}^n \ketbra{v_k}{w_k}$ where $v_k \in \sH_0$ and $w_k \in \sK_0$ for all $k$. Then
\begin{align*}
\cG_{T_0,T_0} &= \sum_{j,k=1}^n \cG_{\ketbra{v_j}{w_j},\ketbra{v_k}{w_k}} = \sum_{j,k =1}^n E_{v_j,v_k} \times F_{w_k,w_j}.
\end{align*} 
Since $E_{v_j,v_k} \ll E_{v,v}$ and $F_{w_k,w_j} \ll F_{w,w}$ for all $j$ and $k$, we have $E_{v_j,v_k} \times F_{w_k, w_j} \ll E_{v,v} \times F_{w,w}$ for all $j$ and $k$ by Lemma \ref{lem:product_meas_absolute_continuity}, hence $\cG_{T_0,T_0} \ll E_{v,v} \times F_{w,w}$. Finally, since $T$ is the limit in $\HS(\sK, \sH)$ of operators $T_0$ as above, we have that $\cG_{T,T}$ is the limit of measures $\cG_{T_0, T_0}$ with $T_0$ as above, so it follows that $\cG_{T,T} \ll E_{v,v} \times F_{w,w}$.
\end{proof}

\section{Integral Projective Tensor Products}\label{sec:integral_projective_tensor_products}

We retain the notation of the previous section. In particular, we let $\cG:\Sigma_{X \times Y} \to \cB(\HS(\sK, \sH))$ be the projection-valued measure constructed in Theorem \ref{thm:double_PVM}.

Let $(Z, \Sigma_Z)$ be a measurable space. Given functions $\alpha:X \times Z \to \bbC$ and $\beta:Y \times Z \to \bbC$ and a point $z \in Z$,  we define $\alpha_z:X \to \bbC$ by $\alpha_z(x) = \alpha(x,z)$ and we define $\beta_z$ similarly. By slight abuse of notation we also write $\alpha_z \beta_z$ for the function $X \times Y \to \bbC$ defined by $(\alpha_z\beta_z)(x,y) = \alpha(x,z)\beta(y,z)$. If $\alpha$ and $\beta$ are measurable, then each function $\alpha_z$, $\beta_z$, and $\alpha_z\beta_z$ is measurable.

Now let $\nu$ be a positive measure on $(Z, \Sigma_Z)$ and let $\alpha:X \times Z \to \bbC$ and $\beta:Y \times Z \to \bbC$ be measurable functions. We will be considering integrals of the form
\begin{equation}\label{eq:norm_integrals}
	N_{v,w} \defeq \int_Z \inorm{\alpha_z}{E}{v} \inorm{\beta_z}{F}{w} \dd{\nu(z)},
\end{equation}
where $v \in \sH$, $w \in \sK$, and for ease of notation we write
\[
\norm{\cdot}_v \defeq \norm{\cdot}_{L^\infty(E_{v,v})} \qqtext{and} \norm{\cdot}_w \defeq \norm{\cdot}_{L^\infty(F_{w,w})}.
\]
Likewise, we will abbreviate
\[
	\norm{\cdot}_T \defeq \norm{\cdot}_{L^\infty(\cG_{T,T})}.
\]
Let us first show that the integrand above is measurable. The proposition below is used tacitly in many works on double operator integrals and is even explicitly pointed out in a few places, but we do not know where to find a proof in the literature. We found the proof below on MathOverflow \cite{381225}.

\begin{proposition}\label{prop:esssup_measurable}
Let $(W, \Sigma_W, \mu)$ be a $\sigma$-finite positive measure space and let $(Z, \Sigma_Z)$ be a measurable space. If $f:W \times Z \to \bbC$ is measurable, then the function
\[
	g:Z \to [0,\infty], \quad g(z) = \norm{f_z}_{L^\infty(\mu)}
\]
is measurable.
\end{proposition}

\begin{proof}
Fix $a > 0$. It suffices to show that $\qty{z \in Z: g(z) > a}$ is measurable. The set $M = \qty{(w,z) \in W \times Z: \abs{f(w,z)} > a}$ is measurable, so by Tonelli's theorem we know the function 
\[
	h:Z \to [0,\infty], \quad h(z) = \int_W \chi_{M}(w,z) \dd{\mu(w)}
\]
is measurable. For each $z \in Z$, define 
\[
M_z = \qty{w \in W: (w,z) \in M} = f_z^{-1}\qty(\qty{b \in \bbC: \abs{b} > a}).
\] 
Then we have
\[
	h(z) = \int_W \chi_{M_z} \dd{\mu} = \mu(M_z).
\]

Finally, we claim that 
\[
	\qty{z \in Z: g(z) > a} = \qty{z \in Z: h(z) > 0},
\] 
which is measurable since $h$ is measurable.  If $g(z) > a$, then there exists $b_0$ in the essential range of $f_z$ such that $\abs{b_0} > a$. Since $\qty{b \in \bbC: \abs{b} > a}$ is an open set containing $b_0$, this implies 
\[
	\mu(M_z) = \mu(f^{-1}_z(\qty{b \in \bbC: \abs{b} > a})) > 0.
\]
If $g(z) \leq a$, then the essential range $R$ of $f_z$ is contained in the closed disk $D_a$ of radius $a$, hence $\bbC \setminus D_a \subset \bbC \setminus R$. Since $f_z^{-1}(\bbC \setminus R)$ has measure zero, we see that
\[
	\mu(M_z) = \mu(f^{-1}_z(\bbC \setminus D_a)) = 0.
\]
This proves the claim.
\end{proof}

We will now consider the supremum of the integrals \eqref{eq:norm_integrals} across all $v \in \sH$ and $w \in \sK$. We show one way this supremum can be rewritten, but first we provide the following elementary fact of measure theory. Below, we let $\essran f$ denote the essential range of a function $f$.

\begin{lemma}\label{lem:separated_esssup}
Let $\lambda$ and $\mu$ be $\sigma$-finite positive measure spaces on $X$ and $Y$, respectively. Let $f:X \to \bbC$ and $g:Y \to \bbC$ be measurable and let $h:X \times Y \to \bbC$ be the measurable function $h(x,y) = f(x)g(y)$. Then
\begin{equation}\label{eq:essran_product}
\essran h \supset \qty{ab : a \in \essran f, \, b \in \essran g}
\end{equation}
and
\begin{equation}\label{eq:separated_esssup}
	\norm{h}_{L^\infty(\lambda \times \mu)} = \norm{f}_{L^\infty(\lambda)}\norm{g}_{L^\infty(\mu)}.
\end{equation}
\end{lemma}

\begin{proof}
Let $a \in \essran f$ and $b \in \essran g$. Let $W$ be a neighborhood of $ab$. There exist neighborhoods $U$ of $a$ and $V$ of $b$ such that $U \cdot V \subset W$. We see that $f^{-1}(U) \times g^{-1}(V) \subset h^{-1}(W)$. Since $f^{-1}(U)$ and $g^{-1}(V)$ have positive measure, we see that $h^{-1}(W)$ has positive measure. Thus, $ab \in \essran h$.

From this it is clear that the right hand side of \eqref{eq:separated_esssup} is less than or equal to the left hand side. If either $\lambda$ or $\mu$ is identically zero, then $\lambda \times \mu$ is identically zero, hence both sides of \eqref{eq:separated_esssup} are zero. Suppose neither $\lambda$ nor $\mu$ is identically zero. Observe that 
\[
	h^{-1}(\bbC \setminus \qty{0}) = f^{-1}(\bbC \setminus \qty{0}) \times g^{-1}(\bbC \setminus \qty{0}).
\]
If $\norm{f}_{L^\infty(\lambda)} = 0$ or $\norm{g}_{L^\infty(\mu)} = 0$, then $f^{-1}(\bbC \setminus \qty{0})$ or $g^{-1}(\bbC \setminus \qty{0})$ has measure zero, respectively, so $h^{-1}(\bbC \setminus \qty{0})$ has measure zero, hence the essential range of $h$ is $\qty{0}$, so $\norm{h}_{L^\infty(\lambda \times \mu)} = 0$. 

Suppose that $\norm{f}_{L^\infty(\lambda)}$ and $\norm{g}_{L^\infty(\mu)}$ are both strictly positive. If either is infinite, then it is clear that $\essran h$ is unbounded and therefore both sides of \eqref{eq:separated_esssup} are infinite. If both are finite, then $\essran f$ and $\essran g$ are both compact, hence the right hand side of \eqref{eq:essran_product} is compact, hence closed. We observe that
\[
	h^{-1}\qty(\bbC \setminus \essran f \cdot \essran g) \subset X \times g^{-1}(\bbC \setminus \essran g) \cup f^{-1}(\bbC \setminus \essran f) \times Y.
\]
The set on the right has measure zero. Thus, $\essran h \subset \essran f \cdot \essran g$ in this case. It is now clear that the left hand side of \eqref{eq:separated_esssup} is less than or equal to the right hand side.
\end{proof}

\begin{lemma}\label{lem:esssup_GTT_dominated_vw}
If $(Z, \Sigma_Z)$ is a measurable space and $\alpha:X \times Z \to \bbC$ and $\beta:Y \times Z \to \bbC$ are measurable, then for every $T \in \HS(\sK, \sH)$ there exist $v \in \sH$ and $w \in \sK$ such that
\[
	\inorm{\alpha_z\beta_z}{\cG}{T} \leq \inorm{\alpha_z}{E}{v} \inorm{\beta_z}{F}{w}
\]
for all $z \in Z$.
\end{lemma}

\begin{proof}
By Theorem \ref{thm:G_TT_absolute_cont_product}, there exist vectors $v \in \sH$ and $w \in \sK$ such that $\cG_{T,T} \ll E_{v,v} \times F_{w,w}$. Thus,
\begin{align*}
	\inorm{\alpha_z\beta_z}{\cG}{T} &\leq \inorm{\alpha_z \beta_z}_{L^\infty(E_{v,v} \times F_{w,w})}  = \inorm{\alpha_z}{E}{v} \inorm{\beta_z}{F}{w},
\end{align*}
as desired, where we used Lemma \ref{lem:separated_esssup} in the second step.
\end{proof}

\begin{proposition}
Let $(Z, \Sigma_Z, \nu)$ be a positive measure space and let $\alpha:X \times Z \to \bbC$ and $\beta: Y \times Z \to \bbC$ be measurable functions. With $N_{v,w}$ as defined in \eqref{eq:norm_integrals}, we have
\begin{equation}\label{eq:sup_vw_T_equal}
	\sup_{v \in \sH, w \in \sK} N_{v,w} = \sup_{T \in \HS(\sK, \sH)} \int_Z \inorm{\alpha_z \beta_z}{\cG}{T} \dd{\nu(z)}.
\end{equation}
\end{proposition}

\begin{proof}
Given $v \in \sH$ and $w \in \sK$, set $T = \ketbra{v}{w}$. Then
\begin{align*}
	\inorm{\alpha_z \beta_z}{\cG}{T} &= \norm{\alpha_z \beta_z}_{L^\infty(E_{v,v}\times F_{w,w})} = \inorm{\alpha_z}{E}{v} \inorm{\beta_z}{F}{w}.
\end{align*}
This shows that the left hand side of \eqref{eq:sup_vw_T_equal} is no greater than the right hand side. The other direction is provided by Lemma \ref{lem:esssup_GTT_dominated_vw}.
\end{proof}

We now show that this supremum is achieved. Moreover, we show that there exist vectors that maximize the integrand for almost every $z$.

\begin{lemma}\label{lem:esssup_ae_dominator}
Let $(Z, \Sigma_Z,\nu)$ be a positive measure space and  $\alpha:X \times Z \to \bbC$  a measurable function. 
\begin{enumerate}[label=\tn{(\alph*)}]
	\item There exists $v_0 \in \sH$ such that
\[
\int_Z \inorm{\alpha_z}{E}{v_0} \dd{\nu(z)} = \sup_{v \in \sH} \int_Z \inorm{\alpha_z}{E}{v}\dd{\nu(z)}.
\]
	\item If $v_0$ is as above and the above quantity is finite, then for every $v \in \sH$ we have
\begin{equation}\label{eq:Linfinity_norm_dominated}
	\inorm{\alpha_z}{E}{v} \leq \inorm{\alpha_z}{E}{v_0}
\end{equation}
for $\nu$-\ae$ z \in Z$.
	\item If $\nu$ is $\sigma$-finite and $\alpha$ is bounded, then there exists $v_0 \in \sH$ such that for every $v \in \sH$ we have \eqref{eq:Linfinity_norm_dominated} for $\nu$-\ae $z \in Z$.
\end{enumerate}
\end{lemma}

\begin{proof}
(a) Define
\[
	N_v \defeq \int_Z \inorm{\alpha_z}{E}{v} \dd{\nu(z)} \qqtext{and} N \defeq \sup_{v \in \sH} N_v.
\]
Let $(v_n)_{n \in \bbN}$  be a sequence in $\sH$ such that $N_{v_n} \to N$. By Theorem \ref{thm:sep_spectral_dominator} we know there exists $v_0 \in \sH$ such that $E_{v_n, v_n} \ll E_{v_0, v_0}$  for all $n \in \bbN$. Then for each $z \in Z$ it follows that the essential range of $\alpha_z$ with respect to $E_{v_n, v_n}$ is contained in the essential range with respect to $E_{v_0, v_0}$, hence 
\begin{equation}\label{eq:v_norm_inequality}
	\inorm{\alpha_z}{E}{v_n} \leq \inorm{\alpha_z}{E}{v_0}.
\end{equation}
It follows that $N_{v_0} = N$.

(b) Let $v_0$ be as in part (a) and assume $N < \infty$. Let $v \in \sH$ be arbitrary. Suppose that the set
\[
	M \defeq \qty{z \in Z : \inorm{\alpha_z}{E}{v} > \inorm{\alpha_z}{E}{v_0}}
\]
has positive measure. Again by Theorem \ref{thm:sep_spectral_dominator} there exists $v' \in \sH$ such that $E_{v,v} \ll E_{v', v'}$ and $E_{v_0, v_0} \ll E_{v', v'}$ and we get corresponding inequalities similar to \eqref{eq:v_norm_inequality}. Now we have
\begin{align*}
	N_{v_0} &= \int_{M^c} \inorm{\alpha_z}{E}{v_0} \dd{\nu(z)} + \int_{M} \inorm{\alpha_z}{E}{v_0} \dd{\nu(z)}\\
	&< \int_{M^c} \inorm{\alpha_z}{E}{v_0} \dd{\nu(z)} + \int_M \inorm{\alpha_z}{E}{v} \dd{\nu(z)}\\
	&\leq N_{v'},
\end{align*}
where the strict inequality follows from the fact that $M$ has positive measure. But this contradicts that $N_{v_0} = N$. Therefore $M$ has measure zero.

(c) Suppose $\alpha$ is bounded and $\nu$ is $\sigma$-finite and let $(Z_n)$ be a partition of $Z$ into countably many measurable sets of finite measure. Then for every $n$ we know
\[
	\sup_{v \in \sH} \int_{Z_n} \inorm{\alpha_z}{E}{v} \dd{\nu(z)} < \infty.
\]
By parts (a) and (b) there exists $v_n \in \sH$ such that for every $v \in \sH$ the set
\[
	W_{n,v} \defeq \qty{z \in Z_n: \inorm{\alpha_z}{E}{v} > \inorm{\alpha_z}{E}{v_n}}
\]
has measure zero. By Theorem \ref{thm:sep_spectral_dominator} we know there exists $v_0 \in \sH$ such that $E_{v_n, v_n} \ll E_{v_0, v_0}$ for all $n$. Now fix $v \in \sH$ and let $W_v = \bigcup_n W_{n,v}$. Then $W_v$ has measure zero and if $z \in Z \setminus W_v$, then $z \in Z_n \setminus W_{n,v}$ for some $n$ and
\[
	\inorm{\alpha_z}{E}{v} \leq \inorm{\alpha_z}{E}{v_n} \leq \inorm{\alpha_z}{E}{v_0},
\]
as desired.
\end{proof}

\begin{lemma}\label{lem:tildes}
Let $(Z, \Sigma_Z, \nu)$ be a $\sigma$-finite positive measure space and let $\alpha:X \times Z \to \bbC$ and $\beta:Y \times Z \to \bbC$ be bounded measurable functions such that
\[
	N_{v,w} \defeq \int_Z \inorm{\alpha_z}{E}{v} \inorm{\beta_z}{F}{w} \dd{\nu(z)} < \infty
\]
for all $v \in \sH$ and $w \in \sK$. 
\begin{enumerate}[label=\tn{(\alph*)}]
	\item There exist $v_0 \in \sH$ and $w_0 \in \sK$ such that for all $T \in \HS(\sK, \sH)$, we have
	\begin{equation}\label{eq:normalization_ineq1}
		\inorm{\alpha_z \beta_z}{\cG}{T} \leq \inorm{\alpha_z}{E}{v_0}\inorm{\beta_z}{F}{w_0}
	\end{equation}
	for $\nu$-\ae $z \in Z$. In particular, for all $v \in \sH$ and $w \in \sK$, we have
	\begin{equation}\label{eq:normalization_ineq2}
	\inorm{\alpha_z}{E}{v}\inorm{\beta_z}{F}{w} \leq \inorm{\alpha_z}{E}{v_0} \inorm{\beta_z}{F}{w_0}
	\end{equation}
	for $\nu$-\ae $z \in Z$.
	\item The function $z \mapsto \alpha(x,z)\beta(y,z)$ is integrable with respect to $\nu$ for $\cG$-\ae $(x,y) \in X \times Y$.
	\item There exists a finite positive measure $\tilde \nu$ on $(Z, \Sigma_Z)$ and measurable functions $\tilde \alpha:X \times Z \to \bbC$ and $\tilde \beta: Y \times Z \to \bbC$ such that $\abs*{\tilde \alpha(x,z)} \leq 1$ and $\abs*{\tilde\beta(y,z)} \leq 1$, such that
	\begin{equation}\label{eq:Nvw_tilde}
		N_{v,w} = \int_Z \inorm{\tilde\alpha_z}{E}{v}\inorm{\tilde\beta_z}{F}{w} \dd{\tilde\nu(z)}
	\end{equation}
	for all $v \in \sH$ and $w \in \sK$, and such that
	\begin{equation}\label{eq:integral_tilde}
		\int_Z \alpha(x,z)\beta(y,z) \dd{\nu(z)} = \int_Z \tilde \alpha(x,z)\tilde \beta(y,z) \dd{\tilde \nu(z)}
	\end{equation}
	for all $\cG$-\ae $(x,y) \in X \times Y$.
\end{enumerate}
\end{lemma}

\begin{proof}
(a) By Lemma \ref{lem:esssup_ae_dominator} (with $E$ replaced by $\cG$ and $\alpha_z$ replaced by $\alpha_z\beta_z$), there exists $T_0 \in \HS(\sK, \sH)$ such that for all $T \in \HS(\sK, \sH)$, we have
\[
	\inorm{\alpha_z\beta_z}{\cG}{T} \leq \inorm{\alpha_z\beta_z}{\cG}{T_0}
\]
for $\nu$-\ae $z \in Z$. By Lemma \ref{lem:esssup_GTT_dominated_vw}, there exist vectors $v_0 \in \sH$ and $w_0 \in \sK$ such that 
\[
	\inorm{\alpha_z \beta_z}{\cG}{T_0} \leq \inorm{\alpha_z}{E}{v_0}\inorm{\beta_z}{F}{w_0},
\]
for every $z \in Z$. Combining the above two equations yields \eqref{eq:normalization_ineq1}. Given $v \in \sH$ and $w \in \sK$, setting $T = \ketbra{v}{w}$ and combining the above two equations yields \eqref{eq:normalization_ineq2}. 

(b) By Tonelli's theorem, the function $h:X \times Y \to [0,\infty]$ defined by
\[
	h(x,y) = \int_Z \abs{\alpha(x,z)\beta(y,z)} \dd{\nu(z)}
\]
is measurable. Let $v_0$ and $w_0$ be as in part (a). Given $T \in \HS(\sK, \sH)$, Tonelli's theorem gives
\begin{align*}
\int_{X \times Y} h \dd{\cG_{T,T}} &= \int_Z \int_{X \times Y} \abs{\alpha_z(x)\beta_z(y)} \dd{\cG_{T,T}(x,y)} \dd{\nu(z)}\\
&\leq \norm{T}^2_{\HS}\int_Z \inorm{\alpha_z \beta_z}{\cG}{T} \dd{\nu(z)}\\
&\leq \norm{T}^2_{\HS} \int_Z \inorm{\alpha_z}{E}{v_0}\inorm{\beta_z}{F}{w_0} \dd{\nu(z)} \\
&<\infty.
\end{align*}
This implies that $\cG(\qty{h = \infty}) = 0$ (cf.~\cite[Thm.~13.25]{grandparudin}).

(c) Since $\nu$ is $\sigma$-finite and $\alpha$ is bounded, Lemma \ref{lem:esssup_ae_dominator} yields $v_0 \in \sH$ such that for every $v \in \sH$ we have
\[
	\inorm{\alpha_z}{E}{v} \leq \inorm{\alpha_z}{E}{v_0}
\]
for $\nu$-\ae $z \in Z$. Similarly, Lemma \ref{lem:esssup_ae_dominator} yields $w_0 \in \sK$ such that for every $w \in \sK$ we have
\[
	\inorm{\beta_z}{F}{w} \leq \inorm{\beta_z}{F}{w_0}
\]
for $\nu$-\ae $z \in Z$.

Define $\tilde \nu$ by
\[
	\dd{\tilde \nu(z)} = \inorm{\alpha_z}{E}{v_0} \inorm{\beta_z}{F}{w_0} \dd{\nu(z)}.
\]
Then $\tilde \nu$ is a finite positive measure on $(Z, \Sigma_Z)$. Define $\tilde \alpha : X \times Z \to \bbC$ by
\[
	\tilde \alpha(x,z) = \begin{cases} \alpha(x,z)/\inorm{\alpha_z}{E}{v_0} &\tn{if } 0 < \abs{\alpha(x,z)} \leq \inorm{\alpha_z}{E}{v_0} \\ 0 & \tn{otherwise} \end{cases}.
\]
Then $\tilde \alpha$ is a measurable function with $\abs{\tilde \alpha(x,z)} \leq 1$ for all $(x,z) \in X \times Z$. Define $\tilde \beta$ similarly.  Note that
\begin{equation}\label{eq:tilde_alpha_norm_cases}
	\tilde \alpha(x,z)\inorm{\alpha_z}{E}{v_0} = \begin{cases} \alpha(x,z) &\tn{if }  \abs{\alpha(x,z)} \leq \inorm{\alpha_z}{E}{v_0}\\
	0 &\tn{if } \inorm{\alpha_z}{E}{v_0} < \abs{\alpha(x,z)} \end{cases}
\end{equation}
and similarly for $\tilde \beta$. 

Now we prove \eqref{eq:Nvw_tilde}. Fix $v \in \sH$ and $w \in \sK$. Let 
\begin{equation}\label{eq:Z_0}
	Z_0 = \qty{z \in Z:  \inorm{\alpha_z}{E}{v_0} < \inorm{\alpha_z}{E}{v} \tn{ or } \inorm{\beta_z}{F}{w_0} < \inorm{\beta_z}{F}{w} }.
\end{equation}
By definition of $v_0$ and $w_0$ we know $\nu(Z_0) = 0$. If $z \in Z_0^c$, then \eqref{eq:tilde_alpha_norm_cases} implies that $\tilde \alpha(x,z)\inorm{\alpha_z}{E}{v_0} = \alpha(x,z)$ for $E_{v,v}$-\ae $x \in X$, hence
\[
	\inorm{\tilde \alpha_z}{E}{v}\inorm{\alpha_z}{E}{v_0} = \inorm{\alpha_z}{E}{v}.
\]
Similar statements apply to $\tilde \beta$ and \eqref{eq:Nvw_tilde} now follows.

Finally, we prove \eqref{eq:integral_tilde}. Define $\varphi:X \times Y \to \bbC$ by
\[
\varphi(x,y) = \int_{Z} \alpha(x,z)\beta(y,z) \dd{\nu(z)}
\]
for all $(x,y) \in X \times Y$ such that $z \mapsto \alpha(x,z)\beta(y,z)$ is integrable and define $\varphi(x,y) = 0$ for all other $(x,y) \in X \times Y$. Then $\varphi$ is measurable and the calculation in part (b) shows that $\varphi$ is essentially bounded with respect to $\cG$. The right hand side of \eqref{eq:integral_tilde} is a bounded measurable function of $(x,y) \in X \times Y$. Therefore it suffices to show that for every $T \in \HS(\sK, \sH)$, we have
\begin{align*}
	\int_{X \times Y} \varphi \dd{\cG_{T,T}} &= \int_{X \times Y} \int_Z \tilde \alpha(x,z)\tilde \beta(y,z) \dd{\tilde \nu(z)} \dd{\cG_{T,T}(x,y)}\\
	&= \int_{X \times Y} \int_Z \tilde \alpha(x,z)\inorm{\alpha_z}{E}{v_0} \tilde \beta(y,z)\inorm{\beta_z}{F}{w_0} \dd{\nu(z)} \dd{\cG_{T,T}(x,y)}.
\end{align*}

Fix $T \in \HS(\sK, \sH)$. Using Theorem \ref{thm:G_TT_absolute_cont_product} we may choose $v \in \sH$ and $w \in \sK$ such that $\cG_{T,T} \ll E_{v,v} \times F_{w,w}$. Let $Z_0$ be as in \eqref{eq:Z_0} and recall that $\nu(Z_0) = 0$. By Fubini's theorem, it suffices to show
\begin{equation}\label{eq:fubini_final_WTS}
\begin{aligned}
\int_{Z_0^c} &\int_{X \times Y} \alpha(x,z)\beta(y,z) \dd{\cG_{T,T}(x,y)}\dd{\nu(z)} \\
&= \int_{Z_0^c} \int_{X \times Y} \tilde \alpha(x,z) \inorm{\alpha_z}{E}{v_0} \tilde \beta(y,z)\inorm{\beta_z}{F}{w_0} \dd{\cG_{T,T}(x,y)}\dd{\nu(z)}.
\end{aligned}
\end{equation}
For fixed $z \in Z_0^c$, the sets
\begin{align*}
	M_{z} &= \qty{x \in X: \inorm{\alpha_z}{E}{v} < \abs{\alpha(x,z)}}\\
	N_z &= \qty{y \in Y: \inorm{\beta_z}{F}{w} < \abs{\beta(y,z)}}
\end{align*}
have measure zero with respect to $E_{v,v}$ and $F_{w,w}$, so $M_z \times Y \cup X \times N_z$ has measure zero with respect to $\cG_{T,T}$. 
But for $z \in Z_0^c$ and $(x,y) \in M_z^c \times N_z^c$, we have 
\[
	\abs{\alpha(x,z)} \leq \inorm{\alpha_z}{E}{v} \leq \inorm{\alpha_z}{E}{v_0}
\]
so that $\tilde \alpha(x,z)\inorm{\alpha_z}{E}{v_0} = \alpha(x,z)$ by \eqref{eq:tilde_alpha_norm_cases}. Similar remarks apply to $\beta$, so the integrands in \eqref{eq:fubini_final_WTS} agree $\cG_{T,T}$-\ae for every $z \in Z_0^c$, verifying \eqref{eq:fubini_final_WTS}.
\end{proof}

\begin{definition}\label{def:function_space}
Define $\fM_{00}$ to be the set of functions $\phi:X \times Y \rightarrow \bbC$ for which there exists a finite positive measure space $(Z, \Sigma_Z, \nu)$ and bounded measurable functions $\alpha:X \times Z \rightarrow \bbC$ and $\beta:Y \times Z \rightarrow \bbC$ such that
\begin{equation}\label{eq:decomposition_def}
	\phi(x,y) = \int_Z \alpha(x,z) \beta(y,z) \dd{\nu(z)},
\end{equation}
for all $(x, y) \in X \times Y$. We call the data $(Z, \Sigma_Z, \nu, \alpha, \beta)$ a \emph{strict decomposition} of $\phi$. It is clear that every $\phi \in \fM_{00}$ is bounded. Furthermore, as part of Fubini's theorem we know every $\phi \in \fM_{00}$ is measurable. Thus, $\fM_{00} \subset \B(X \times Y)$.

The set $\fM_{00}$ is defined without reference to any measure on $X$, $Y$, or $X \times Y$. We would, however, like to ``neglect'' sets of measure zero with respect to $\cG$. Therefore, we let $\pi:\B(X \times Y) \to L^\infty(\cG)$ be the canonical $*$-homomorphism defined by quotienting by functions that are zero almost everywhere with respect to $\cG$, and we define
\[
	\fM \defeq \pi(\fM_{00}).
\] 
It is also helpful to define
\[
	\fM_0 \defeq \pi^{-1}(\fM).
\]
Thus, an element of $\fM_0$ is a bounded measurable function $\phi$ for which there exists a finite positive measure space $(Z, \Sigma_Z, \nu)$ and bounded measurable functions $\alpha:X \times Z \to \bbC$ and $\beta:Y \times Z \to \bbC$ such that \eqref{eq:decomposition_def} holds $\cG$-\ae When \eqref{eq:decomposition_def} holds $\cG$-\ae rather than everywhere, we call $(Z, \Sigma_Z, \nu, \alpha, \beta)$ a \emph{decomposition} of $\phi$ rather than a strict decomposition.

For every $\phi \in \fM_0$, we define 
\begin{equation}\label{eq:M_norm}
	\norm{\phi}_{\fM_0} = \inf_{(Z, \Sigma_Z, \nu, \alpha, \beta)} \sup_{v, w} \int_Z \inorm{\alpha_z}{E}{v} \inorm{\beta_z}{F}{w} \dd{\nu(z)},
\end{equation}
where the infimum is taken over all decompositions of $\phi$ and the supremum is taken over all $v \in \sH$ and $w \in \sK$.
\end{definition}

\begin{remark}\label{rem:complex_measure}
Equivalent definitions of $\fM_{00}$, $\fM_0$, and $\fM$ are obtained if we allow $\nu$ to be a complex measure, as seen by using the polar decomposition of $\nu$. Likewise, in the definition of $\norm{\phi}_{\fM_0}$ one may take the infimum over all decompositions of $\phi$ with $\nu$ as a complex measure if one also takes the integral in \eqref{eq:M_norm} with respect to the total variation $\abs{\nu}$ instead of $\nu$.
\end{remark}

\begin{remark}\label{rem:easy_norm_bounds}
If $(Z, \Sigma_Z, \nu, \alpha, \beta)$ is a decomposition of $\phi \in \fM_0$, then a simple bound on $\norm{\phi}_{\fM_0}$ is given by
\[
	\norm{\phi}_{\fM_0} \leq \norm{\alpha}_{\B(X \times Z)} \norm{\beta}_{\B(Y \times Z)} \nu(Z).
\]
If $z \mapsto \norm{\alpha_z}_{\B(X)}$ and $z \mapsto \norm{\beta_z}_{\B(Y)}$ happen to be measurable, then we also have the bound
\[
	\norm{\phi}_{\fM_0} \leq \int_Z \norm{\alpha_z}_{\B(X)} \norm{\beta_z}_{\B(Y)} \dd{\nu(z)}.
\]
\end{remark}

\begin{remark}\label{rem:separable_case}
If $\sH$ and $\sK$ are separable, then Theorem \ref{thm:sep_spectral_dominator} implies that there exists $v_0 \in \sH$ and $w_0 \in \sK$ such that $E_{v_0,v_0}$ and $F_{w_0,w_0}$ have the same sets of measure zero as $E$ and $F$, respectively. In this case, for any decomposition $(Z, \Sigma_Z, \nu, \alpha, \beta)$ of $\phi \in \fM_0$, we have $\norm{\alpha_z}_v \leq \norm{\alpha_z}_{v_0} = \norm{\alpha_z}_{L^\infty(E)}$ and similarly for $\beta$, hence
\[
	\sup_{v,w}\int_Z \inorm{\alpha_z}{E}{v} \inorm{\beta_z}{F}{w} \dd{\nu(z)} = \int_Z \norm{\alpha_z}_{L^\infty(E)} \norm{\beta_z}_{L^\infty(F)} \dd{\nu(z)}.
\]
The right hand side is the integrand most commonly seen in the literature.
\end{remark}

Below we elucidate the structure of these sets and the function $\norm{\cdot}_{\fM_0}$. The proof will make use of the construction of a disjoint union of positive measure spaces. By definition, if $(Z_i, \Sigma_i, \nu_i)_{i \in I}$ is a family of positive measure spaces indexed by a set $I$, then we may equip the disjoint union $Z \defeq \bigsqcup_{i \in I} Z_i$ with the $\sigma$-algebra 
\begin{equation}\label{eq:C_disjoint_union}
	\Sigma \defeq \qty{E \subset Z: E \cap Z_i \in \Sigma_i \textnormal{ for all $i \in I$}}
\end{equation}
and the measure $\nu: \Sigma \rightarrow [0,\infty]$ defined by
\begin{equation}\label{eq:nu_disjoint_union}
	\nu(E) = \sum_{i \in I} \nu_i(E \cap Z_i).
\end{equation}

\begin{theorem}[{cf.~\cite[Prop.~4.1.4]{Nikitopolous}, \cite[Lem.~4.6]{PagterSukochev}}]\label{thm:integral_projective_tensor_product}
The set $\fM_{00}$ is a $*$-subalgebra of $\B(X \times Y)$ containing the constant function $\phi(x,y) = 1$. Hence, $\fM$ is a unital $*$-subalgebra of $L^\infty(\cG)$ and $\fM_0$ is a unital $*$-subalgebra of $\B(X \times Y)$. The function $\norm{\cdot}_{\fM_0}$ is a submultiplicative seminorm on $\fM_0$ invariant under complex conjugation and further satisfies
\begin{equation}\label{eq:norm_comparison}
	\norm{\pi(\phi)}_{L^\infty(\cG)} \leq \norm{\phi}_{\fM_0}
\end{equation}
for every $\phi \in \fM_0$. Thus, $\norm{\phi}_{\fM_0} = 0$ if and only if $\pi(\phi) = 0$, so $\norm{\cdot}_{\fM_0}$ descends to a well-defined norm $\norm{\cdot}_{\fM}$ on $\fM$. Furthermore, $\fM$ is a unital Banach $*$-algebra when equipped with $\norm{\cdot}_{\fM}$.
\end{theorem}

\begin{proof}
Let $\phi \in \fM_{00}$ and let $(Z, \Sigma_Z, \nu, \alpha, \beta)$ be a strict decomposition of $\phi$. Then it is clear that $(Z, \Sigma_Z, \nu, \alpha^*, \beta^*)$ is a strict decomposition of $\phi^*$, so $\phi^* \in \fM_{00}$. If $\lambda \in \bbC$, then it is clear that $(Z, \Sigma_Z, \nu, \lambda \alpha, \beta)$ is a strict decomposition of $\lambda \phi$, so $\lambda \phi \in \fM_{00}$.

We show $\fM_{00}$ is closed under addition and multiplication. For $i \in \qty{1,2}$, let $\phi_i \in \fM_{00}$ and let $(Z_i, \Sigma_{Z_i}, \nu_i, \alpha_i, \beta_i)$ be a strict decomposition of $\phi_i$. 

To show $\phi_1 + \phi_2 \in \fM_{00}$, let $(Z, \Sigma_Z, \nu)$ be the disjoint union of $(Z_1, \Sigma_{Z_1}, \nu_1)$ and $(Z_2, \Sigma_{Z_2}, \nu_2)$. Define $\alpha:X \times Z \rightarrow \bbC$ by
\[
	\alpha(x,z) = \begin{cases} \alpha_1(x,z) &: z \in Z_1 \\ \alpha_2(x,z) &: z \in Z_2 \end{cases}
\]
and define $\beta:Y \times Z \rightarrow \bbC$ similarly. Then $\alpha$ and $\beta$ are bounded measurable functions and
\begin{align*}
	\int_Z \alpha(x,z) \beta(y,z) \dd{\nu(z)} &= \sum_{i=1}^2 \int_{Z_i} \alpha_i(x,z)\beta_i(y,z) \dd{\nu_i(z)} \\
	&= \sum_{i=1}^2 \phi_i(x,y).
\end{align*}
This proves that $\phi_1 + \phi_2 \in \fM_{00}$.

To show $\phi_1 \cdot \phi_2 \in \fM_{00}$, let $(Z, \Sigma_Z, \nu)$ be the product of $(Z_1, \Sigma_{Z_1}, \nu_1)$ and $(Z_2, \Sigma_{Z_2}, \nu_2)$. Define $\alpha:X \times Z \rightarrow \bbC$ by
\[
	\alpha(x,z_1, z_2) = \alpha_1(x,z_1)\alpha_2(x,z_2)
\]
and define $\beta:Y \times Z \rightarrow \bbC$ similarly.  Then $\alpha$ and $\beta$ are bounded and measurable and Fubini's theorem gives
\begin{align*}
\int_Z \alpha(x,z_1, z_2) \beta(y,z_1,z_2) \dd{\nu(z_1, z_2)} &= \prod_{i=1}^2 \int_{Z_i} \alpha_i(x,z_i) \beta_i(y, z_i) \dd{\nu_i(z_i)}\\
&= \prod_{i=1}^2 \phi_i(x,y).
\end{align*}
Thus, $\phi_1 \cdot \phi_2 \in \fM_{00}$. 

Thus, $\fM_{00}$ is a $*$-subalgebra of $\B(X \times Y)$. Since $\fM$ is the image of $\fM_{00}$ under the $*$-homomorphism $\pi$, we know $\fM$ is a $*$-subalgebra of $L^\infty(\cG)$. Since $\fM_0$ is the preimage of $\fM$ under the $*$-homomorphism $\pi$, we know $\fM_0$ is a $*$-subalgebra of $\B(X \times Y)$.

Note that if $\phi, \phi_1, \phi_2 \in \fM_0$ and the decompositions above of $\phi$, $\phi_1$, and $\phi_2$ are not required to be strict, then the constructions above produce (non-strict) decompositions of $\phi^*$, $\lambda \phi$, $\phi_1+\phi_2$, and $\phi_1 \cdot \phi_2$. We will use these constructions to verify the properties of $\norm{\cdot}_{\fM_0}$.

Using the decompositions above, we have:
\begin{align*}
\norm{\phi^*}_{\fM_0} &\leq \sup_{v, w} \int_Z \inorm{\alpha_z^*}{E}{v} \inorm{\beta_z^*}{F}{w} \dd{\nu(z)} \\
&= \sup_{v,w} \int_Z \inorm{\alpha_z}{E}{v} \inorm{\beta_z}{F}{w} \dd{\nu(z)}.
\end{align*}
Since the decomposition of $\phi$ was arbitrary, we see that $\norm{\phi^*}_{\fM_0} \leq \norm{\phi}_{\fM_0}$. Replacing $\phi$ with $\phi^*$ yields the reverse inequality, so $\norm{\phi}_{\fM_0} = \norm{\phi^*}_{\fM_0}$.

It is clear that $\norm{0}_{\fM_0} = 0$, hence $\norm{\lambda \phi}_{\fM_0} = \abs{\lambda} \norm{\phi}_{\fM_0}$ if $\lambda = 0$. If $\lambda \neq 0$, then we have
\begin{align*}
	\norm{\lambda \phi}_{\fM_0} &\leq \sup_{v,w} \int_Z \inorm{\lambda \alpha_z}{E}{v} \inorm{\beta_z}{F}{w} \dd{\nu(z)} \\
	&= \abs{\lambda} \sup_{v,w} \int_Z \inorm{\alpha_z}{E}{v} \inorm{\beta_z}{F}{w} \dd{\nu(z)}.
\end{align*}
Since the decomposition was arbitrary, we have $\norm{\lambda \phi}_{\fM_0} \leq \abs{\lambda} \norm{\phi}_{\fM_0}$. Replacing $\lambda$ with $\lambda^{-1}$ and $\phi$ with $\lambda \phi$ yields the reverse inequality, so we have homogeneity of the seminorm.

Next, with the decomposition of $\phi_1 + \phi_2$ constructed above, we compute:
\begin{align*}
\norm{\phi_1+\phi_2}_{\fM_0} &\leq \sup_{v,w} \int_Z \inorm{\alpha_z}{E}{v} \inorm{\beta_z}{F}{w} \dd{\nu(z)}\\
&\leq \sum_{i=1}^2 \sup_{v,w} \int_{Z_i} \inorm{\alpha_{i,z}}{E}{v} \inorm{\beta_{i,z}}{F}{w} \dd{\nu(z)}.
\end{align*}
The triangle inequality follows.

With the decomposition of $\phi_1 \cdot \phi_2$ constructed above, we have:
\begin{align*}
\norm{\phi_1\cdot\phi_2}_{\fM_0} &\leq \sup_{v,w} \int_{Z_1 \times Z_2} \inorm{\alpha_{1,z_1}\alpha_{2,z_2}}{E}{v} \inorm{\beta_{1,z_1}\beta_{2,z_2}  }{F}{w} \dd{\nu(z_1,z_2)}\\
&\leq \prod_{i=1}^2 \sup_{v,w} \int_{Z_i} \inorm{\alpha_{i,z_i}}{E}{v} \inorm{\beta_{i,z_i}}{F}{w} \dd{\nu_i(z_i)}.
\end{align*}
It follows that $\norm{\cdot}_{\fM_0}$ is submultiplicative.

Finally, note that $\norm{\pi(\phi)}_{L^\infty(\cG)}$ is the operator norm of the normal operator  $\int_{X \times Y} \phi \dd{\cG}$. Thus, using Fubini's theorem and Theorem \ref{thm:G_TT_absolute_cont_product}, we compute:
\begin{align*}
\norm{\pi(\phi)}_{L^\infty(\cG)} &= \underset{\norm{T}_{\HS} \leq 1}{\sup_{T \in \HS(\sK, \sH)}} \abs{\int_{X \times Y} \phi \dd{\cG_{T,T}}}\\
&\leq \underset{\norm{T}_{\HS} \leq 1}{\sup_{T \in \HS(\sK, \sH)}} \int_Z \int_{X \times Y} \abs{\alpha(x,z)\beta(y,z)} \dd{\cG_{T,T}(x,y)} \dd{\nu(z)}\\
&\leq \underset{\norm{T} \leq 1}{\sup_{T \in \HS(\sK, \sH)}} \int_Z \inorm{\alpha_z\beta_z}{\cG}{T} \dd{\nu(z)}\\
&\leq \sup_{v,w} \int_Z \inorm{\alpha_z}{E}{v} \inorm{\beta_z}{F}{w} \dd{\nu(z)}.
\end{align*}
From this it follows that $\norm{\pi(\phi)}_{L^\infty(\cG)} \leq \norm{\phi}_{\fM_0}$ since the decomposition was arbitrary.

If $\phi \in \fM_0$ and $\pi(\phi) = 0$, then letting $(Z, \Sigma_Z, \nu)$ be any finite positive measure space and letting $\alpha$ and $\beta$ be identically zero gives a decomposition of $\phi$, from which we see that $\norm{\phi}_{\fM_0} = 0$. Conversely, if $\norm{\phi}_{\fM_0} = 0$, then $\norm{\pi(\phi)}_{L^\infty(\cG)} = 0$, so $\pi(\phi) = 0$. Thus, $\norm{\cdot}_{\fM_0}$ descends to a well-defined norm $\norm{\cdot}_{\fM}$ on $\fM$. Furthermore, $\norm{\cdot}_{\fM}$ is submultiplicative, invariant under complex conjugation, and dominates $\norm{\cdot}_{L^\infty(\cG)}$.

By setting $(Z, \Sigma_Z, \nu)$ to be any positive measure space with $\nu(Z) = 1$ and setting $\alpha(x,z) = \beta(y,z) = 1$ for all $x \in X$, $y  \in Y$, and $z \in Z$, we find that $(Z, \Sigma_Z, \nu, \alpha, \beta)$ is a strict decomposition for the constant function $\phi(x,y) = 1$, so this function is in $\fM_{00}$. From the definition of $\norm{\cdot}_{\fM_0}$ and the strict decompositions above we see that $\norm{\phi}_{\fM_0} \leq 1$. Since $\norm{\pi(\phi)}_{L^\infty(\cG)} = 1$, \eqref{eq:norm_comparison} gives $\norm{\phi}_{\fM_0} = 1$.

All that remains to do is show that $\fM$ is complete. Let $(\pi(\phi_n))_{n \in \bbN}$ be a Cauchy sequence in $\fM$, where $\phi_n \in \fM_0$ for each $n \in \bbN$. Then the sequence $(\pi(\phi_n))_{n \in \bbN}$ is Cauchy in $L^\infty(\cG)$ and therefore has a limit $\pi(\phi) \in L^\infty(\cG)$, where $\phi \in B(X \times Y)$. Furthermore, by passing to a subsequence, we may assume that $\phi_n(x,y) \to \phi(x,y)$ for $\cG$-\ae $(x,y) \in X \times Y$.\footnote{That this may be done for Cauchy sequences in $L^p(\mu)$, $1 \leq p \leq \infty$, and $\mu$ a positive measure is proven in \cite[Thms.~3.11 \& 3.12]{paparudin}. When $p = \infty$, an identical proof works for projection-valued measures.} By passing to another subsequence we may assume that
\[
	\norm{\phi_{n+1} - \phi_n}_{\fM_0} < \frac{1}{2^n}
\]
for all $n \in \bbN$. Let $\phi_0 = 0$ for convenience. For each $n \in \bbN$ we may choose a decomposition $(Z_n, \Sigma_{Z_n}, \nu_n, \alpha_n, \beta_n)$ of $\phi_{n+1} - \phi_n$  such that
\[
	\sup_{v, w} \int_{Z_n} \inorm{\alpha_{n,z}}{E}{v}\inorm{\beta_{n,z}}{F}{w} \dd{\nu_n(z)} < \frac{1}{2^n}.
\]
By scaling the $\alpha_n$, $\beta_n$, and $\nu_n$, we may assume that $\abs{\alpha_n(x,z)} \leq 1$ and $\abs{\beta_n(y,z)} \leq 1$ for all $x \in X$, $y \in Y$, and $z \in Z$. 

Fix $m \in \qty{0} \cup \bbN$. 
Let $(Z, \Sigma_Z, \nu)$ be the disjoint union of the measure spaces $(Z_n, \Sigma_{Z_n}, \nu_n)$ for $n \geq m$. Note that $\nu$ is $\sigma$-finite. Define $\alpha:X \times Z \to \bbC$ by $\alpha(x,z) = \alpha_n(x,z)$ where $n$ is the unique integer such that $z \in Z_n$. Define $\beta:Y \times Z \to \bbC$ similarly. Then $\alpha$ and $\beta$ are bounded and measurable and for every $v \in \sH$ and $w \in \sK$ we have
\[
	\int_Z \inorm{\alpha_z}{E}{v}\inorm{\beta_z}{F}{w} \dd{\nu(z)} = \sum_{n=m}^\infty \int_{Z_n} \inorm{\alpha_{n,z}}{E}{v}\inorm{\beta_{n,z}}{F}{w} \dd{\nu_n(z)} < \infty.
\]
Let $\tilde \nu$, $\tilde \alpha$, and $\tilde \beta$ be the measure and functions obtained by applying part (c) of Lemma \ref{lem:tildes} to $\nu$, $\alpha$, and $\beta$. Then for $\cG$-\ae $(x,y) \in X \times Y$, we have
\begin{align*}
\phi(x,y) - \phi_m(x,y) &= \sum_{n=m}^\infty \phi_{n+1}(x,y) - \phi_n(x,y)\\
&= \sum_{n=m}^\infty \int_{Z_n} \alpha_n(x,z)\beta_n(y,z) \dd{\nu_n(z)}\\
&= \int_Z \alpha(x,z)\beta(y,z) \dd{\nu(z)}\\
&= \int_Z \tilde \alpha(x,z)\tilde \beta(y,z) \dd{\tilde \nu(z)}.
\end{align*}
Thus, $(Z, \Sigma_Z, \tilde \nu, \tilde \alpha, \tilde \beta)$ is a decomposition of $\phi - \phi_m$. Setting $m = 0$ proves that $\phi \in \fM$. For arbitrary $m$, we have
\begin{align*}
\norm{\phi - \phi_m}_{\fM_0} &\leq \sup_{v, w} \int_Z \inorm{\tilde \alpha_z}{E}{v}\inorm{\tilde \beta_z}{F}{w} \dd{\tilde \nu(z)}\\
&= \sup_{v,w} \int_Z \inorm{\alpha_z}{E}{v}\inorm{\beta_z}{F}{w} \dd{\nu(z)}\\
&\leq  \sum_{n=m}^\infty \sup_{v,w} \int_{Z_n} \inorm{\alpha_{n,z}}{E}{v}\inorm{\beta_{n,z}}{F}{w} \dd{\nu_n(z)}\\
&< \frac{1}{2^{m-1}}.
\end{align*}
This proves that $\pi(\phi_m) \to \pi(\phi)$ with respect to $\norm{\cdot}_{\fM}$, so $\fM$ is complete.
\end{proof}

Henceforth we will leave the projection map $\pi:B(X \times Y) \to L^\infty(\cG)$  implicit, as is conventional.
If $\phi \in \fM$, then one can understand the matrix elements  $\ev{S,\int \phi \dd{\cG}T}_{\HS}$ in terms of the projection-valued measures $E$ and $F$.

\begin{theorem}\label{thm:HS_mel_decomposition}
If $\phi \in \fM$ and $(Z, \Sigma_Z, \nu, \alpha, \beta)$ is a decomposition of $\phi$, then for any $S, T \in \HS(\sK, \sH)$ the map
\begin{equation}\label{eq:measurable_mel_decomposition}
	z \mapsto \Tr(S^* \int_X \alpha_z \dd{E}  T  \int_Y \beta_z \dd{F})
\end{equation}
is bounded and measurable and
\begin{equation}\label{eq:mel_integral_expansion}
	\Tr(S^* \int_{X \times Y} \phi \dd{\cG} T) = \int_Z \Tr(S^* \int_X \alpha_z \dd{E} T \int_Y \beta_z \dd{F}) \dd{\nu(z)}.
\end{equation}
\end{theorem}

\begin{proof}
It is clear that \eqref{eq:measurable_mel_decomposition} is bounded. It suffices to show measurability when $S$ and $T$ are rank-one operators. Indeed,  this will imply measurability when $S$ and $T$ are finite-rank operators by sesquilinearity of the inner product. For arbitrary Hilbert--Schmidt operators $S$ and $T$ we can then take sequences $(S_n)$ and $(T_n)$ of finite-rank operators converging in $\HS(\sK, \sH)$ to $S$ and $T$  and obtain \eqref{eq:measurable_mel_decomposition} as the pointwise limit of a sequence of measurable functions.

Thus, let $v_1, v_2 \in \sH$, let $w_1, w_2 \in \sK$, let $S = \ketbra{v_1}{w_1}$ and let $T = \ketbra{v_2}{w_2}$. Then
\begin{align*}
\ev{S, \int_X \alpha_z \dd{E}  T  \int_Y \beta_z \dd{F}}_{\HS} &= \int_X \alpha_z \dd{E_{v_1, v_2}} \int_Y \beta_z \dd{F_{w_2, w_1}}
\end{align*}
which is measurable by Fubini's theorem. 

Moreover, by Propositions \ref{prop:curly_E_F} and \ref{prop:totally_separated_variable_integral} we have
\begin{align*}
\ev{S, \int_X \alpha_z \dd{E}  T  \int_Y \beta_z \dd{F}}_{\HS} &= \int_{X \times Y} \alpha_z \beta_z \dd{\cG_{S, T}}.
\end{align*}
Finally, integrating over $Z$ and using Fubini's theorem yields \eqref{eq:mel_integral_expansion} for the case when $S$ and $T$ are rank-one operators. Then \eqref{eq:mel_integral_expansion} follows when $S$ and $T$ are finite-rank operators by sesquilinearity of both sides. For arbitrary Hilbert--Schmidt operators $S$ and $T$, \eqref{eq:mel_integral_expansion} follows by taking sequences $(S_n)$ and $(T_n)$ of finite-rank operators converging to $S$ and $T$ in $\HS(\sK, \sH)$ and using the dominated convergence theorem.
\end{proof}

\section{The Double Operator Integral on \texorpdfstring{$\TC(\sK, \sH)$}{TC} and \texorpdfstring{$\fB(\sK, \sH)$}{B}} \label{sec:double_operator_integral_on_B}

We would now like to extend \eqref{eq:mel_integral_expansion} to the case where $T \in \TC(\sK, \sH)$ and $S \in \fB(\sK, \sH)$. To show the left hand side of \eqref{eq:mel_integral_expansion} is well-defined in this case, we first show that if $\phi \in \fM$, then $\int \phi \dd{\cG}$ maps trace-class operators to trace-class operators. To show the right hand side is well-defined, we will have to show that the integrand is measurable. Note that the absolute value of the integrand on the right-hand side is bounded above by $\norm{S^*}\norm{\alpha}_{\B(X \times Z)} \norm{T}_{\TC} \norm{\beta}_{\B(Y \times Z)}$, so measurability will imply integrability. 

The results of this section, excluding Proposition \ref{prop:SOT_sequential_continuity}, are similar to the main results in \cite[\S4.2-4.3]{Nikitopolous}. We have streamlined the proofs, however.

\begin{theorem}\label{thm:TC_to_TC}
If $\phi \in \fM$, then $\int \phi \dd{\cG}$ maps trace-class operators to trace-class operators and is bounded as an operator $\TC(\sK, \sH) \to \TC(\sK,\sH)$ with norm
\[
	\norm{\int_{X \times Y} \phi \dd{\cG}}_{\TC \to \TC} \leq \norm{\phi}_\fM.
\]
\end{theorem}

\begin{proof}
Let $T \in \TC(\sK, \sH)$. The theorem will be proven if we can show that for any finite orthonormal families $v_1,\ldots, v_n \in \sH$ and $w_1, \ldots, w_n \in \sK$, we have 
\begin{equation}\label{eq:trace_class_test}
	\sum_{i = 1}^n \abs{\ev{v_i, \qty(\int \phi \dd{\cG})(T) w_i}} \leq \norm{\phi}_\fM \norm{T}_{\TC}.
\end{equation}
Let $\sH_0$ be the closed subspace of $\sH$ generated by $\qty{v_1, \ldots, v_n} \cup T(\sK)$ and let $\sK_0$ be the closed subspace of $\sK$ generated by $\qty{w_1,\ldots, w_n} \cup (\ker T)^\perp$. Note that $\sH_0$ and $\sK_0$ are separable. Let $P:\sH \to \sH_0$ and $Q:\sK \to \sK_0$ be the  projections onto these subspaces. 

Let $(Z, \Sigma_Z, \nu, \alpha, \beta)$ be a decomposition of $\phi$. Let $\mathrm{LHS}$ be the left hand side of \eqref{eq:trace_class_test} for brevity. By Theorem \ref{thm:HS_mel_decomposition} with $S = \ketbra{v_i}{w_i}$, we have
\begin{align*}
\mathrm{LHS} &= \sum_{i=1}^n\abs{\int_Z \ev{v_i, \int_X \alpha_z \dd{E} T \int_Y \beta_z \dd{F} w_i} \dd{\nu(z)}}\\
&\leq \int_Z \sum_{i=1}^n \abs{ \ev{v_i, P\int_X \alpha_z \dd{E}P^*P T Q^*Q \int_Y \beta_z \dd{F} Q^* w_i} } \dd{\nu(z)}.
\end{align*}
The integrand is bounded above by
\[
	\norm{P\int_X \alpha_z \dd{E} P^*}\norm{PTQ^*}_\TC  \norm{Q\int_Y \beta_z \dd{F}Q^*}
\]
and we also have $\norm{PTQ^*}_{\TC} \leq \norm{T}_\TC$. By Theorem \ref{thm:sep_spectral_dominator}, there exists $v \in \sH$ such that $E_{u_1, u_2} \ll E_{v,v}$ for all $u_1, u_2 \in \sH_0$.  Thus,
\begin{align*}
\norm{P\int_X \alpha_z \dd{E}P^*} &= \underset{\norm{u_1}, \norm{u_2} \leq 1}{\sup_{u_1, u_2 \in \sH_0}} \abs{\int_X \alpha_z \dd{E_{u_1, u_2}}} \leq \norm{\alpha_z}_v.
\end{align*}
Similarly, there exists $w \in \sK$ such that 
\[
	\norm{Q \int_Y \beta_z \dd{F}Q^*} \leq \norm{\beta_z}_w.
\]
Thus, we have
\[
	\mathrm{LHS} \leq \norm{T}_{\TC} \int_Z \norm{\alpha_z}_v \norm{\beta_z}_w \dd{\nu(z)}.
\]
The inequality \eqref{eq:trace_class_test} now follows by definition of $\norm{\phi}_\fM$.
\end{proof}

We now show that the integrand on the right hand side of \eqref{eq:mel_integral_expansion} is measurable.

\begin{proposition}\label{prop:ultraweakly_meas}
Let $(Z, \Sigma_Z)$ be a measurable space, let $\alpha:X \times Z \to \bbC$ and $\beta:Y \times Z \to \bbC$ be bounded measurable functions, and let $S \in \cB(\sK, \sH)$, and let $T \in \TC(\sK, \sH)$. The function
\begin{equation}\label{eq:B_TC_trace_func}
	z \mapsto \Tr\qty(S^*\int_X \alpha_z \dd{E} T \int_Y \beta_z \dd{F})
\end{equation}
is measurable.
\end{proposition}

\begin{proof}
It suffices to prove measurability of \eqref{eq:B_TC_trace_func} when $T$ is a rank-one operator. The result will then follow for finite-rank $T$ by linearity, and the result will follow for arbitrary trace-class $T$ by approximating with a sequence of finite-rank operators and realizing \eqref{eq:B_TC_trace_func} as the pointwise limit of a sequence of measurable functions.

Fix $v \in \sH$ and $w \in \sK$ and set $T = \ketbra{v}{w}$. Then \eqref{eq:B_TC_trace_func} becomes
\[
	z \mapsto \ev{w, \int_Y \beta_z \dd{F} S^* \int_X \alpha_z \dd{E} v}.
\]
Since $\alpha$ and $\beta$ can be approximated uniformly by sequences of simple functions, the above function can be written as a pointwise limit of linear combinations of functions of the form
\begin{equation}\label{eq:char_ultraweakly_meas}
	z \mapsto \ev{w, \int_Y \chi_N(y,z) \dd{F(y)} S^* \int_X \chi_M(x,z) \dd{E(x)}v},
\end{equation}
where $M \in \Sigma_{X \times Z}$ and $N \in \Sigma_{Y \times Z}$. It therefore suffices to show that each function of the form  \eqref{eq:char_ultraweakly_meas} is measurable.

Suppose $M = M_1 \times M_2$ for measurable sets $M_1 \subset X$ and $M_2 \subset Z$. If also $N = N_1 \times N_2$ where $N_1 \subset Y$ and $N_2 \subset Z$ are measurable, then \eqref{eq:char_ultraweakly_meas} reduces to $z \mapsto \ev{w, F(N_1)S^*E(M_1)v} \chi_{M_2 \cap N_2}(z)$, which is clearly measurable. Consider the set $\Lambda$ of all sets $N \in \Sigma_{Y \times Z}$ such that \eqref{eq:char_ultraweakly_meas} is measurable, still assuming $M = M_1 \times M_2$. Then $\Lambda$ contains every measurable rectangle and it is clear that $\Lambda$ is closed under taking complements. If $(N_i)_{i \in \bbN}$ is a disjoint collection of elements of $\Lambda$ and $N = \bigcup_{i \in \bbN} N_i$, then
\[
	\int_Y \chi_N(y,z) \dd{F(y)}w = \lim_{n \to \infty} \sum_{i=1}^n \int_Y \chi_{N_i}(y,z) \dd{F(y)}w
\]
from which it follows that \eqref{eq:char_ultraweakly_meas} is measurable as a pointwise limit of measurable functions. Thus, $\Lambda$ is a $\lambda$-system containing the $\pi$-system of measurable rectangles in $Y \times Z$, so $\Lambda = \Sigma_{Y \times Z}$ by the Dynkin $\pi$--$\lambda$ theorem.

We have shown that \eqref{eq:char_ultraweakly_meas} is measurable when $M$ is a measurable rectangle and $N \in \Sigma_{Y \times Z}$ is arbitrary. Fixing an arbitrary $N$, the same argument using the Dynkin $\pi$--$\lambda$ theorem, now applied to $M$, shows that \eqref{eq:char_ultraweakly_meas} is measurable for arbitrary $M \in \Sigma_{X \times Z}$. 
\end{proof}

\begin{theorem}\label{thm:doi_TC_mel}
Let $T \in \TC(\sK, \sH)$, let $S \in \fB(\sK, \sH)$, let $(Z, \Sigma_Z, \nu)$ be a finite positive measure space, and let $\alpha:X \times Z \to \bbC$ and $\beta:Y \times Z \to \bbC$ be bounded measurable functions. Define $\phi:X \times Y \to \bbC$ by 
\[
	\phi(x,y) = \int_Z \alpha(x,z) \beta(y,z) \dd{\nu(z)}.
\]
Then:
\begin{equation}\label{eq:trace_decomposition}
\Tr\qty(S^* \int_{X \times Y} \phi \dd{\cG}T) = \int_Z \Tr\qty(S^*\int_X \alpha_z \dd{E} T \int_Y \beta_z \dd{F}) \dd{\nu(z)}.
\end{equation}
\end{theorem}

\begin{proof}
Fixing $T$, $S$, and $(Z, \Sigma_Z, \nu)$, both the left and right hand sides of \eqref{eq:trace_decomposition} are bounded bilinear functions of $(\alpha, \beta) \in \B(X \times Z) \times B(Y \times Z)$, where the boundedness of the left hand side follows from Theorem \ref{thm:TC_to_TC} and Remark \ref{rem:easy_norm_bounds}.  It therefore suffices to prove \eqref{eq:trace_decomposition} when $\alpha$ and $\beta$ are characteristic functions $\alpha = \chi_M$ and $\beta = \chi_N$ for some measurable sets $M \in \Sigma_{X \times Z}$ and $N \in \Sigma_{Y \times Z}$. We will again use the Dynkin $\pi$--$\lambda$ theorem.

Suppose $M = M_1 \times M_2$ for measurable sets $M_1 \subset X$ and $M_2 \subset Z$.  If $N = N_1 \times N_2$ for measurable sets $N_1 \subset Y$ and $N_2 \subset Z$, then both sides of \eqref{eq:trace_decomposition} reduce to 
\[
	\Tr(S^*E(M_1)TF(N_1))\nu(M_2 \cap N_2).
\]
Consider the set $\Lambda$ of all sets $N \in \Sigma_{Y \times Z}$ such that \eqref{eq:trace_decomposition} holds, still assuming $M = M_1 \times M_2$. Then $\Lambda$ contains every measurable rectangle and it is clear that $\Lambda$ is closed under taking complements. Suppose $(N_i)_{i \in \bbN}$ is a disjoint collection of elements of $\Lambda$ and let $N = \bigcup_{i \in \bbN} N_i$. Define 
\[
	\beta_n(y,z) \defeq \sum_{i=1}^n \chi_{N_i}(y,z) \qqtext{and} \phi_n(x,y) \defeq \int_Z \alpha(x,z)\beta_n(y,z) \dd{\nu(z)}.
\]
Then define:
\[
R_n \defeq S^*\int_{X \times Y} (\phi - \phi_n)\dd{\cG}T.
\]
Each operator $R_n$ is trace-class by Theorem \ref{thm:TC_to_TC}. Then for each $n \in \bbN$ there exists a countable orthonormal family $(e_{i})_{i \in I}$ such that 
\[
	\Tr(R_n) = \sum_{i \in I} \ev{e_{i}, R_n e_{i}}.
\]
Taking absolute values and using Theorem \ref{thm:HS_mel_decomposition}, we estimate:
\begin{align}
\abs{\Tr(R_n)} &\leq \sum_{i \in I} \int_Z  \abs{\ev{e_{i}, S^*\int_X \alpha_z \dd{E} T \int_Y (\beta - \beta_n)_z \dd{F} e_{i}} } \dd{\nu(z) }\\
&= \int_Z \sum_{i \in I} \abs{\ev{e_{i}, S^* \int_X \alpha_z \dd{E} T \int_Y(\beta - \beta_n)_z \dd{F} e_{i}} } \dd{\nu(z)}, \label{eq:sum_int_Fubini}
\end{align}
where the interchange of sum and integral is justified since the index set $I$ is countable. The integrand is bounded above by
\begin{align*}
\norm{S^* \int_X \alpha_z \dd{E} T \int_Y (\beta - \beta_n)_z \dd{F}}_{\TC} &\leq \norm{S^*}\norm{T \int_Y (\beta - \beta_n)_z \dd{F}}_{\TC}\\
&\leq \norm{S^*} \norm{T}_{\TC}.
\end{align*}
Finally,  note that
\begin{align*}
\norm{T\int_Y(\beta - \beta_n)_z \dd{F}}_{\TC} &\leq \norm{\abs{T}^{1/2}}_{\HS} \norm{\int_Y(\beta - \beta_n)_z \dd{\cF} \abs{T}^{1/2} }_{\HS},
\end{align*}
as can be proven by identifying $\TC(\sK, \sH)^* \cong \fB(\sH, \sK)$ and using the polar decomposition of $T$.
Since $\beta_n \to \beta$ pointwise and $\beta - \beta_n$ is uniformly bounded above by $1$, we know the above converges to $0$ as $n \to \infty$. Thus, the integrand of \eqref{eq:sum_int_Fubini} converges pointwise to $0$ and is uniformly bounded above by $\norm{S^*}\norm{T}_\TC$; by the dominated convergence theorem we have
\[
	\lim_{n \to \infty} \abs{\Tr(R_n)} = 0.
\]
Therefore the left hand side of \eqref{eq:trace_decomposition} is the limit of the same expression with $\phi$ replaced by $\phi_n$.

We must now show that the right hand side of \eqref{eq:trace_decomposition} is the limit of the same expression with $\beta$ replaced by $\beta_n$. The idea is essentially the same. First observe that
\begin{align*}
	\lim_{n \to \infty} \abs{T}^{1/2} \int_Y (\beta_n)_z \dd{F} &= \lim_{n \to \infty} \int_Y (\beta_n)_z \dd{\cF} \abs{T}^{1/2}  \\
	&= \int_Y \beta_z \dd{\cF} \abs{T}^{1/2}  = \abs{T}^{1/2} \int_Y \beta_z \dd{F},
\end{align*}
where the convergence is with respect to the Hilbert--Schmidt norm. Using the polar decomposition of $T$, it follows immediately that
\[
	\lim_{n \to \infty} \Tr(S^* \int_X \alpha_z \dd{E} T \int_Y (\beta_n)_z \dd{F}) = \Tr(S^* \int_X \alpha_z \dd{E} T \int_Y \beta_z \dd{F}).
\]
As functions of $z$, the expressions above are bounded above by the constant function with value $\norm{S^*}\norm{T}_{\TC}$, so the dominated convergence theorem implies that the limit above commutes with the integral over $Z$, which is the desired result.

We have proven that $\Lambda$ is a $\lambda$-system containing the $\pi$-system $\Sigma_Y \times \Sigma_Z$, and therefore $\Lambda = \Sigma_{Y \times Z}$ by the Dynkin $\pi$--$\lambda$ theorem. We conclude that \eqref{eq:trace_decomposition} holds whenever $\alpha$ is the characteristic function of a measurable rectangle and $\beta$ is any characteristic function. Now fixing $\beta$ to be an arbitrary characteristic function, a similar application of the Dynkin $\pi$--$\lambda$ theorem applied to $\alpha$ proves that \eqref{eq:trace_decomposition} holds whenever $\alpha$ is an arbitrary characteristic function. Thus, \eqref{eq:trace_decomposition} holds in generality.
\end{proof}

Given $\phi \in \fM$, we can use Theorems \ref{thm:TC_to_TC} and \ref{thm:doi_TC_mel} to extend the definition of $\int_{X \times Y} \phi \dd{\cG} T$ to the case where $T \in \fB(\sK, \sH)$. First we establish some notation.

\begin{definition}
Let $\widetilde \cG:\Sigma_{Y \times X} \to \fB(\HS(\sH, \sK))$ be the unique projection-valued measure such that
\[
	\widetilde \cG(\Delta \times \Gamma)R = F(\Delta)R E(\Gamma)
\]
for all $\Delta \in \Sigma_Y$, $\Gamma \in \Sigma_X$, and $R \in \HS(\sH, \sK)$. Let $\tau:Y \times X \to X \times Y$ be the map $\tau(y,x) = (x,y)$. Note that the formula
\[
	\cG_{S,T}(\Lambda) = \widetilde \cG_{T^*, S^*}(\tau^{-1}(\Lambda))
\]
is easily checked for all measurable rectangles $\Lambda$ and $S, T \in \HS(\sK, \sH)$, and therefore holds for all measurable sets $\Lambda \in \Sigma_{X \times Y}$. In particular, if $\Lambda \in \Sigma_{X \times Y}$ and $\cG(\Lambda) = 0$, then $\widetilde \cG(\tau^{-1}(\Lambda)) = 0$. 

Let $\widetilde \fM_{00}$, $\widetilde \fM_0$, and $\widetilde \fM$ be the spaces as defined in Definition \ref{def:function_space}, but constructed using $\widetilde \cG$ instead of $\cG$. Given $\phi:X \times Y \to \bbC$, let $\widetilde \phi = \phi \circ \tau$, so
\[
\widetilde \phi(y,x) = \phi(x,y)
\]
for all $(y,x) \in Y \times X$. We observe that $\phi \in \fM_{0}$ (resp.\ $\phi \in \fM_{00}$) with $(Z, \Sigma_Z, \nu, \alpha, \beta)$ as a decomposition (resp.\ a strict decomposition) if and only if $\widetilde \phi \in \widetilde \fM_0$ (resp.\ $\widetilde \phi \in \widetilde \fM_{00}$) with $(Z, \Sigma_Z, \nu, \beta, \alpha)$ as a decomposition (resp.\ a strict decomposition). Thus, the map $\phi \mapsto \widetilde \phi$ defines a $*$-isomorphism $\fM_{00} \to \widetilde \fM_{00}$ and defines an isometric $*$-isomorphism $\fM_0 \to \widetilde \fM_{0}$ that induces a well-defined isometric $*$-isomorphism $\fM \to \widetilde \fM$.
\end{definition}

\begin{theorem}\label{thm:doi_bounded_T}
If $\phi \in \fM$ and $T \in \fB(\sK, \sH)$, then there exists a unique operator $\int_{X \times Y} \phi \dd{\cG}T \in \fB(\sK, \sH)$ such that
\begin{equation}\label{eq:doi_bounded_T_def}
	\Tr(R\int_{X \times Y} \phi \dd{\cG} T) = \Tr(T \int_{Y \times X} \widetilde \phi \dd{\widetilde \cG} R)
\end{equation}
holds for all $R \in \TC(\sH, \sK)$. Furthermore,
\begin{equation}\label{eq:bounded_T_norms}
	\norm{\int_{X \times Y} \phi \dd{\cG} T} \leq \norm{\phi}_\fM \norm{T}.
\end{equation}
For any decomposition $(Z, \Sigma_Z, \nu, \alpha, \beta)$ of $\phi$ and $S \in \TC(\sK, \sH)$, we have
\begin{equation}\label{eq:doi_mel_bounded_T}
\Tr(S^* \int_{X \times Y} \phi \dd{\cG} T) = \int_Z \Tr(S^* \int_X \alpha_z \dd{E} T \int_Y \beta_z \dd{F} ) \dd{\nu(z)}.
\end{equation}
If $T \in \HS(\sK, \sH)$, then this operator agrees with $\int_{X \times Y} \phi \dd{\cG} T$ as originally defined. 
\end{theorem}

\begin{proof}
By Theorem \ref{thm:TC_to_TC}, the linear map
\[
	\TC(\sH, \sK) \to \bbC, \quad R \mapsto \Tr\qty(T \int_{Y \times X} \widetilde \phi \dd{\widetilde \cG} R)
\]
is bounded.  From the isometric isomorphism $\TC(\sH, \sK)^* \cong \fB(\sK, \sH)$, we see that there exists a unique operator $\int \phi \dd{\cG}T \in \fB(\sK, \sH)$ such that \eqref{eq:doi_bounded_T_def} holds 
for all $R \in \TC(\sH, \sK)$. Furthermore, this operator satisfies
\begin{align*}
	\norm{\int_{X \times Y} \phi \dd{\cG}T} = \underset{\norm{R}_\TC \leq 1}{\sup_{R \in \TC(\sH, \sK)}} \abs{\Tr\qty(T \int_{Y \times X} \widetilde \phi \dd{\widetilde \cG} R)}.
\end{align*}
Using Theorem \ref{thm:TC_to_TC} and the fact that $\norm*{\widetilde \phi}_{\widetilde \fM} = \norm{\phi}_\fM$, we obtain
\[
	\norm{\int_{X \times Y} \phi \dd{\cG}T} \leq \norm{T}\norm{\phi}_\fM.
\]

If $(Z, \Sigma_Z, \nu, \alpha, \beta)$ is any decomposition of $\phi$, then $(Z, \Sigma_Z, \nu, \beta, \alpha)$ is a decomposition of $\widetilde \phi$, so setting $R = S^*$ in \eqref{eq:doi_bounded_T_def} and applying Theorem \ref{thm:doi_TC_mel} yields \eqref{eq:doi_mel_bounded_T} 
for all $S \in \TC(\sK, \sH)$. From \eqref{eq:doi_mel_bounded_T} and Theorem \ref{thm:HS_mel_decomposition} we see that this definition of $\int \phi \dd{\cG}T$ agrees with the original definition when $T \in \HS(\sK, \sH)$. 
\end{proof}

Note that in the special case where $S = \ketbra{v}{w}$ for some $v \in \sH$ and $w \in \sK$, \eqref{eq:doi_mel_bounded_T} reduces to:
\begin{equation}\label{eq:doi_simple_mel_bounded_T}
	\ev{v, \int \phi \dd{\cG}T\,w} = \int_Z \ev{v, \int_X \alpha_z \dd{E} T \int_Y \beta_z \dd{F} w} \dd{\nu(z)}.
\end{equation}
The operator $\int \phi \dd{\cG} T$ in Theorem \ref{thm:doi_bounded_T} is our final construction of a double operator integral. It may be denoted as
\[
	\iint \phi(x, y) \dd{E(x)} T \dd{F(y)} \defeq \int \phi \dd{\cG} T 
\]
when convenient. It is also clear from \eqref{eq:doi_simple_mel_bounded_T} that if $\phi(x,y) = \alpha(x)\beta(y)$ for some bounded measurable functions $\alpha$ and $\beta$, then we may take $Z$ to be a one-point set with measure $\nu(Z) = 1$, whence
\[
	\iint \alpha(x)\beta(y) \dd{E(x)}T \dd{F(y)} = \int \alpha \dd{E} T \int \beta \dd{F}.
\]

\begin{proposition}\label{prop:doi_continuity_B(B(K,H))}
Given $\phi \in \fM$, the map 
\[
	\int_{X \times Y} \phi \dd{\cG} :\fB(\sK, \sH) \to \fB(\sK, \sH)
\] is bounded and linear. Furthermore, the map
\begin{equation}\label{eq:phi_to_integral_map}
	\fM \to \fB(\fB(\sK, \sH)), \quad \phi \mapsto \int_{X \times Y} \phi \dd{\cG}
\end{equation}
is an injective unital algebra homomorphism and
\begin{equation}\label{eq:T_bounded_integral_norm}
	\norm{\int_{X \times Y} \phi \dd{\cG}} \leq \norm{\phi}_\fM
\end{equation}
for all $\phi \in \fM$.
\end{proposition}

\begin{proof}
The linearity of $\int \phi \dd{\cG}$ follows from \eqref{eq:doi_bounded_T_def} and the uniqueness clause of Theorem \ref{thm:doi_bounded_T}.
Boundedness of $\int \phi \dd{\cG}$ is given in \eqref{eq:bounded_T_norms}. Thus, \eqref{eq:phi_to_integral_map} is well-defined. Linearity of \eqref{eq:phi_to_integral_map} again follows from \eqref{eq:doi_bounded_T_def} and the uniqueness clause of Theorem \ref{thm:doi_bounded_T}, and the fact that $\phi \mapsto \widetilde \phi$ is linear. The estimate \eqref{eq:T_bounded_integral_norm} follows immediately from \eqref{eq:bounded_T_norms}.

If $\phi(x,y) = 1$ for all $(x,y) \in X \times Y$, then $\int \widetilde \phi \dd{\widetilde \cG}$ is the identity operator on $\HS(\sH, \sK)$, hence $\int \phi \dd{\cG}$ is the identity operator by uniqueness in Theorem \ref{thm:doi_bounded_T} and the cyclic property of the trace.  

If $\phi \in \fM$ and $\int \phi \dd{\cG} = 0$, then $\int \phi \dd{\cG} T = 0$ for all $T \in \HS(\sK, \sH)$, and this implies that $\phi = 0$ since integration against $\cG$ is injective as a map $L^\infty(\cG) \to \cB(\HS(\sK, \sH))$.

If $\phi_1, \phi_2 \in \fM$, then for all $T \in \fB(\sK, \sH)$ and $R \in \TC(\sH, \sK)$, we have
\begin{align*}
\Tr\qty(R\int \phi_1 \phi_2 \dd{\cG} T) &= \Tr\qty(T \int \widetilde{\phi_1} \widetilde{\phi_2} \dd{\widetilde \cG} R)\\
&= \Tr\qty(T \int \widetilde \phi_1 \dd{\widetilde \cG}\qty(\int \widetilde \phi_2 \dd{\widetilde \cG} R))\\
&= \Tr\qty(\int \widetilde \phi_2 \dd{\widetilde \cG} R \int \phi_1 \dd{\cG} T)\\
&= \Tr\qty(\int \phi_1 \dd{\cG} T \int {\widetilde \phi_2} \dd{\widetilde \cG} R)\\
&= \Tr\qty(R \int \phi_2 \dd{\cG}\qty(\int \phi_1 \dd{\cG} T)).
\end{align*}
Now from uniqueness in Theorem \ref{thm:doi_bounded_T} it follows that \eqref{eq:phi_to_integral_map} respects multiplication.
\end{proof}

\begin{proposition}
If $\phi \in \fM$ and $T \in \fB(\sK, \sH)$, then
\[
	\qty(\int_{X \times Y} \phi \dd{\cG} T)^* = \int_{Y \times X} \widetilde{\phi}^* \dd{\widetilde \cG} T^*.
\]
\end{proposition}

\begin{proof}
Let $v \in \sH$ and $w \in \sK$ and let $(Z, \Sigma_Z, \nu, \alpha, \beta)$ be a decomposition of $\phi$. Then
\begin{align*}
\ev{w, \qty(\int \phi \dd{\cG} T)^* v} &= \ev{v, \qty(\int \phi \dd{\cG} T)  w}^*\\
&= \int_Z \ev{v, \int \alpha_z \dd{E} T \int \beta_z \dd{F} w}^* \dd{\nu(z)}\\
&= \int_Z \ev{w, \int \beta_z^* \dd{F} T^* \int \alpha_z^* \dd{E} v} \dd{\nu(z)}\\
&= \ev{w, \qty(\int \widetilde \phi^* \dd{\widetilde \cG} T^*) v},
\end{align*}
as desired.
\end{proof}

Given $\phi \in \fM$, it is obvious from \eqref{eq:doi_bounded_T_def} that $\int \phi \dd{\cG}:\fB(\sK, \sH) \to \fB(\sK, \sH)$ is continuous with respect to the ultraweak topologies on the domain and codomain. We now show that it is sequentially continuous with respect to the strong operator topologies. This generalizes \cite[Prop.~4.9.i]{AzamovCareyDoddsSukochev}.

\begin{proposition}\label{prop:SOT_sequential_continuity}
Let $\phi \in \fM$. If $(T_n)_{n \in \bbN}$ is a sequence of operators in $\fB(\sK, \sH)$  converging to an operator $T \in \fB(\sK, \sH)$ in the strong operator topology, then $\int \phi \dd{\cG} T_n \to \int \phi \dd{\cG} T$ in the strong operator topology.
\end{proposition}

\begin{proof}
Fix $w \in \sK$ and let
\[
w_n = \qty[\int \phi \dd{\cG}(T_n - T)]w.
\]
We want to show that $w_n \to 0$. First note that since $(T_n - T)_{n \in \bbN}$ is pointwise bounded, the uniform boundedness principle yields $M \geq 0$ such that $\norm{T_n - T} \leq M$ for all $n \in \bbN$. Thus,
\[
	\norm{w_n} \leq \norm{\phi}_\fM M \norm{w}
\]
for all $n \in \bbN$.

Let $(Z, \Sigma_Z, \nu, \alpha, \beta)$ be a decomposition of $\phi$. Observe that
\begin{align*}
\norm{w_n}^2 &= \int_Z \ev{w_n, \int_X \alpha_z \dd{E} (T_n - T) \int_Y \beta_z \dd{F} w} \dd{\nu(z)}.
\end{align*}
The integrand is bounded with bound:
\begin{align*}
	\abs{\tn{integrand}} &\leq \norm{\phi}_\fM M \norm{w} \norm{\alpha}_{\B(X \times Z)} \norm{(T_n - T)\int \beta_z \dd{F} w}\\
	&\leq \norm{\phi}_\fM M^2 \norm{w}^2 \norm{\alpha}_{\B(X \times Z)} \norm{\beta}_{\B(Y \times Z)}.
\end{align*}
From the first line and the strong convergence $T_n \to T$ we see that the integrand converges pointwise to zero as $n \to \infty$. From the second line we see that the integrand is uniformly bounded by a constant, which is integrable since $\nu$ is finite. The dominated convergence theorem now implies that $w_n \to 0$.
\end{proof}

\section{The Daletskii--Krein Formula}\label{sec:DK}

We now give conditions that imply the Daletskii--Krein formula:
\[
	\frac{d}{ds}f(A(s)) = \iint \frac{f(x) - f(y)}{x - y} \dd{E_s(x)} \Phi'(s) \dd{E_s(y)}
\]
for the derivative of a function $f$ of a path of self-adjoint operators $A(s)$. As mentioned in the introduction, this formula and its higher derivative counterparts were the origin of the theory of double operator integrals \cite{DaletskiiKrein}; they have been rederived under various hypotheses and with various technical interpretations by several authors \cite{MR348494,PagterSukochev,nikitopoulos2023higher}. In this section, we prove the Daletskii--Krein formula for an easily tractable class of functions $f$ but for a different type of path $A(s)$ than has previously been considered.

Throughout this section we continue to let $\sH$ be a Hilbert space. We will consider the case where $X = Y = \bbR$, equipped with the Borel $\sigma$-algebra. For any $t \in \bbR$, define $\phi_t :\bbR \times \bbR  \to \bbC$ by 
\[
	\phi_t(x,y) = \begin{cases} \frac{e^{it x} - e^{it y}}{x - y} &: x \neq y \\ it e^{it x} &: x = y \end{cases}.
\]
Then observe that we have a strict decomposition:
\[
	\phi_t(x,y) = \int_0^1 it e^{it x r} e^{it y(1 - r)} \dd{r}
\]
so $\phi_t \in \fM_{00}$. For any two projection-valued measures on $\bbR$, it is clear from the definition of the norm on $\fM$ that
\[
	\norm{\phi_t}_\fM \leq \abs{t}
\]
for all $t \in \bbR$.

\begin{lemma}\label{lem:exponential_difference}
Let $A$ be a densely defined self-adjoint operator on $\sH$, let $\Phi \in \fB(\sH)$ be self-adjoint, and let $B = A + \Phi$. Let $E$ and $F$ be the spectral measures associated to $A$ and $B$, respectively. Then
\begin{equation}\label{eq:exponential_difference}
e^{itB} - e^{itA} = \iint \phi_t(x,y) \dd{E(x)} \Phi \dd{F(y)}
\end{equation} 
for all $t \in \bbR$. In particular,
\[
	\norm{e^{itB} - e^{itA}} \leq \abs{t}\norm{\Phi}.
\]
\end{lemma}

\begin{proof}
Let $C$ be the operator on the right hand side of \eqref{eq:exponential_difference}. Let $\sD$ be the domain of $A$ and note that this is also the domain of $B$.
For any $v, w \in \sD$, we have
\begin{align*}
\ev{v, C w} &= \int_0^1 it \ev{v, \int e^{itxr} \dd{E(x)} \Phi \int e^{ity(1 - r)} \dd{F(y)} w   } \dd{r} \\
&= \int_0^1 it\ev{e^{-itrA}v, B e^{it(1 - r)B}w} - it \ev{e^{-itrA}v, Ae^{it(1 - r)B}w} \dd{r}. \\
&= \int_0^1 \qty(-\frac{d}{dr} \ev{e^{-itrA}v, e^{it(1 - r)B}w}) \dd{r}\\
&= \ev{v, e^{itB}w} - \ev{v, e^{itA}w},
\end{align*}
proving \eqref{eq:exponential_difference}. The estimate on the norm comes from $\norm{\phi_t}_\fM \leq \abs{t}$ and Proposition \ref{prop:doi_continuity_B(B(K,H))}.
\end{proof}

We can extend \eqref{eq:exponential_difference} to functions other than $e^{itx}$ by Fourier analysis. 

\begin{definition}\label{def:Wiener_space}
The \emph{first Wiener space} $W_1(\bbR)$ is defined to be the set of functions $f:\bbR \to \bbC$ for which there exists a complex Borel measure $\mu$ on $\bbR$ such that $\int_\bbR \abs{t} \dd{\abs{\mu}(t)} < \infty$ and
\begin{equation}\label{eq:Fourier_mu}
	f(x) = \int_\bbR e^{itx} \dd{\mu(t)}
\end{equation}
for all $x \in \bbR$.
\end{definition}

For example, if $f :\bbR \to \bbC$ is continuous, integrable and $\int_\bbR \abs*{t \hat f(t)} \dd{t} < \infty$, then by the Fourier inversion formula and the continuity of $f$ we know
\[
	f(x) = \frac{1}{2\pi} \int_\bbR e^{itx} \hat f(t) \dd{t}
\]
for all $x \in \bbR$, hence $f \in W_1(\bbR)$ with the measure $\dd{\mu(t)} = \hat f(t) \dd{t}/2\pi$ satisfying \eqref{eq:Fourier_mu}. Here our convention for the Fourier transform of $f \in L^1(\bbR)$ is
\[
	\hat f(t) := \int_\bbR e^{-itx} f(x) \dd{x}.
\]

In general, if $f \in W_1(\bbR)$ and $\mu$ satisfies \eqref{eq:Fourier_mu}, then it follows from the dominated convergence theorem that $f$ is differentiable and 
\[
	f'(x) = \int_\bbR ite^{itx} \dd{\mu(t)}
\]
for all $x \in \bbR$.

Given $f \in W_1(\bbR)$, we claim the function
\[
	\phi_f(x,y) = \begin{cases} \frac{f(x) - f(y)}{x - y} &: x \neq y \\ f'(x) &: x = y \end{cases}
\]
is in $\fM_{00}$. Observe that for $x \neq y$, we have
\begin{align*}
\frac{f(x) - f(y)}{x - y} &=  \int_\bbR \frac{e^{itx} - e^{ity}}{x - y} \dd{\mu(t)}\\
&= \int_{\bbR}\qty(\int_0^1 it e^{itxr} e^{ity(1 - r)} \dd{r}) \dd{\mu(t)}\\
&= \int_{\bbR \times [0,1]}  e^{itxr}e^{ity(1 - r)} \dd{\nu(t,r)},
\end{align*}
where $\nu$ is the product $\dd{\nu(t,r)} = it\dd{\mu(t)} \times \dd{r}$. When $x = y$, the final line above yields $f'(x)$. Thus, the final line above gives a strict decomposition of $\phi_f$ (see Remark \ref{rem:complex_measure}). It is clear from this decomposition that
\begin{equation}\label{eq:phi_f_norm}
	\norm{\phi_f}_\fM \leq \int_\bbR \abs{t} \dd{\abs{\mu}(t)},
\end{equation}
for any two projection-valued measures on $\bbR$.

\begin{theorem}\label{thm:f_difference}
Let $A$ be a densely defined self-adjoint operator on $\sH$, let $\Phi \in \fB(\sH)$ be self-adjoint, and let $B = A + \Phi$. Let $E$ and $F$ be the spectral measures associated to $A$ and $B$, respectively. If $f \in W_1(\bbR)$, then
\begin{equation}\label{eq:f_difference}
f(B) - f(A) = \iint \phi_f(x,y) \dd{E(x)} \Phi \dd{F(y)}.
\end{equation}
In particular,
\begin{equation}\label{eq:f_difference_norm}
	\norm{f(B) - f(A)} \leq \norm{\Phi} \int_\bbR \abs*{t} \dd{\abs{\mu}(t)}.
\end{equation}
\end{theorem}

\begin{proof}
Let $D$ be the operator on the right hand side of \eqref{eq:f_difference} and let $\mu$ be a complex Borel measure on $\bbR$ such that \eqref{eq:Fourier_mu} holds. Then for any $v, w \in \sH$ we have
\begin{align*}
\ev{v, D w} &= \int_\bbR\int_0^1 it  \ev{v, e^{itrA} \Phi e^{it(1 - r)B} w} \dd{r} \dd{\mu(t)}.
\end{align*}
We see from Lemma \ref{lem:exponential_difference} that  this is
\[
	\ev{v,D w} =\int_\bbR  \ev{v,e^{itB}w} \dd{\mu(t)} - \int_\bbR \ev{v, e^{itA}w} \dd{\mu(t)}.
\]
By Fubini's theorem we have
\begin{align*}
	\int_\bbR \ev{v, e^{itA}w} \dd{\mu(t)} &= \int_\bbR \int_\bbR  e^{itx} \dd{E_{v,w}(x)} \dd{\mu(t)}\\
	&= \int_\bbR f(x) \dd{E_{v,w}(x)} = \ev{v, f(A)w}
\end{align*}
and similarly for $B$. This proves \eqref{eq:f_difference} and \eqref{eq:f_difference_norm} follows from \eqref{eq:phi_f_norm} and Proposition \ref{prop:doi_continuity_B(B(K,H))}.
\end{proof}

\begin{remark}\label{rem:PVM_swap}
An interesting feature of the formulas \eqref{eq:exponential_difference} and \eqref{eq:f_difference} is that the spectral measures can be interchanged. Indeed,
\begin{align*}
	f(B) - f(A) &= -[f(A) - f(B)] \\
	&= -\iint \phi_f(x,y) \dd{F(x)} \qty(-\Phi) \dd{E(y)}\\
	&= \iint \phi_f(x,y) \dd{F(x)} \Phi \dd{E(y)}.
\end{align*}
We will use this version in the proof of the Daletskii--Krein formula.
\end{remark}

To state the kinds of operator-valued maps we will be differentiating in the Daletskii--Krein formula, we make the following definition.

\begin{defn}\label{defn:strongly_diff}
Let $J \subset \bbR$ be an interval. 
We say a map $\Phi:J \to \fB(\sH)$ is \emph{strongly differentiable} if for every $v \in \sH$, the map $s \mapsto \Phi(s)v$ is Fr\'echet differentiable. If $\Phi$ is strongly differentiable, then for each $s \in J$, the linear map $\Phi'(s):\sH \to \sH$ defined by
\[
	\Phi'(s)v = \frac{d}{dt} \Phi(t)v \bigg|_{t = s}
\]
is automatically bounded since it can be seen as a pointwise limit of a sequence of difference quotients $[\Phi(s_n) - \Phi(s)]/(s_n - s)$ where $s_n \to s$, each of which is a bounded linear map. 

Furthermore, if $\Phi:J \to \fB(\sH)$ is strongly differentiable, then in fact $\Phi$ is norm-continuous since for any $s \in J$ and any sequence $(s_n)$ in $J \setminus \qty{s}$ approaching $s$, the difference quotients $[\Phi(s_n) - \Phi(s)]/(s_n - s)$ are pointwise bounded, hence the uniform boundedness principle yields  $M > 0$ such that
\[
	\norm{\frac{\Phi(s_n) - \Phi(s)}{s_n - s}} \leq M
\]
for all $n \in \bbN$. Multiplying by $\abs{s_n - s}$ and taking the limit as $n \to \infty$ shows that $\norm{\Phi(s_n) - \Phi(s)} \to 0$.
\end{defn}

\begin{theorem}[Daletskii--Krein formula]
Let $J \subset \bbR$ be an interval and let $\Phi:J \to \fB(\sH)$ be strongly differentiable and pointwise self-adjoint, i.e., $\Phi(s)^* = \Phi(s)$ for all $s \in J$. Let $A_0$ be a densely defined self-adjoint operator and let $A(s) = A_0 + \Phi(s)$ for all $s \in J$. Let $E_s$ be the spectral measure associated to $A(s)$. If $f \in W_1(\bbR)$, then $t \mapsto f(A(t))$ is strongly differentiable with derivative 
\begin{equation}\label{eq:DK}
	\frac{d}{ds}f(A(s)) = \iint \phi_f(x,y) \dd{E_s(x)} \Phi'(s) \dd{E_s(y)}
\end{equation}
for all $s \in J$. 
\end{theorem}

\begin{proof}
Let $(s_n)$ be a sequence in $J \setminus \qty{s}$ such that $s_n \to s$. For ease of notation, let
\[
	\Delta_n = \frac{\Phi(s_n) - \Phi(s)}{s_n - s}
\]
and let $\cG_n:\Sigma_{\bbR^2} \to \fB(\HS(\sH))$ be the projection-valued measure determined by $\cG_n(\Lambda_1 \times \Lambda_2)T = E_{s_n}(\Lambda_1)TE_s(\Lambda_2)$. 

Let $w \in \sH$. By Theorem \ref{thm:f_difference} and Remark \ref{rem:PVM_swap}, we have
\begin{align*}
\frac{f(A(s_n))w - f(A(s))w}{s_n - s} &= \iint \phi_f \dd{\cG_n}(\Delta_n)w.
\end{align*}
We rewrite the right hand side as
\begin{align*}
\iint \phi_f \dd{\cG_n}(\Delta_n)w &= \iint \phi_f \dd{\cG_n} (\Delta_n - \Phi'(s))w + \iint \phi_f \dd{\cG_n}(\Phi'(s))w.
\end{align*}
We claim that the first term on the right converges to zero and the second term converges to the right hand side of \eqref{eq:DK} applied to $w$.

Let $\mu$ be a complex Borel measure such that \eqref{eq:Fourier_mu} holds. Let $Z = \bbR \times [0,1]$, let $\nu$ be the product measure $\dd{\nu(t,r)} = it\dd{\mu(t)} \times \dd{r}$, let $\alpha(x,t,r) = e^{itrx}$, and let $\beta(y,t,r) = e^{it(1-r)y}$. We have seen that this gives a strict decomposition of $\phi_f$.

Define
\[
	w_n = \iint \phi_f \dd{\cG_n}(\Delta_n - \Phi'(s)) w.
\]
Since $\Phi$ is strongly differentiable, the sequence $\Delta_n - \Phi'(s)$ converges pointwise to zero. By the uniform boundedness principle, there exists $M > 0$ such that
$\norm{\Delta_n - \Phi'(s)} \leq M$
for all $n$. Thus,
\[
	\norm{w_n} \leq M \norm{w}\norm{\phi_f}_\fM \leq  M \norm{w} \int_\bbR \abs{t} \dd{\abs{\mu}(t)}
\]
for all $n \in \bbN$. Then:
\begin{align*}
\norm{w_n}^2 = \int_{\bbR \times [0,1]} \ev{w_n, e^{itrA(s_n)} \qty(\Delta_n -  \Phi'(s)) e^{it(1-r)A(s)}w   } \dd{\nu(t,r)}.
\end{align*}
Since the integrand is bounded above by a constant independent of $n$ and goes to zero pointwise since $\Delta_n - \Phi'(s) \to 0$ strongly, we see that $\norm{w_n}^2 \to 0$  by the dominated convergence theorem. This proves the first part of the claim.

For the second part, let
\[
	v_n := \iint \phi_f \dd{\cG_n}(\Phi'(s))w
\]
and
\[
 v:= \iint \phi_f(x,y) \dd{E_s(x)} \Phi'(s) \dd{E_s(y)}w.
\]
It remains to show that $v_n \to v$. Note that 
\[
	\norm{v_n} \leq \norm{\Phi'(s)} \norm{w} \norm{\phi_f}_\fM \leq \norm{\Phi'(s)} \norm{w} \int_\bbR \abs{t} \dd{\abs{\mu}(t)},
\] 
for all $n$. Now given $n \in \bbN$ and $(t,r) \in \bbR \times [0,1]$, define
\[
u_{n,t,r} = e^{itrA(s_n)}\Phi'(s)e^{it(1-r)A(s)}w.
\]
Observe that $\norm{u_{n,t,r}} \leq \norm{\Phi'(s)}\norm{w}$ for all $n$, $t$, and $r$. Then Lemma \ref{lem:exponential_difference} shows that
\[
	\norm{e^{itrA(s_n)} - e^{itrA(s)}} \leq \abs{tr}\norm{\Phi(s_n) - \Phi(s)}.
\]
As $n \to \infty$, this goes to zero in the norm topology of $\fB(\sH)$ by norm-continuity of $\Phi$. Thus,
\[
	u_{n,t,r} \to u_{t,r} \defeq e^{itrA(s)}\Phi'(s)e^{it(1-r)A(s)}w.
\]

Now, we observe that
\begin{align*}
	\norm{v_n}^2 &= \int_{\bbR \times [0,1]} \ev{v_n, u_{n,t,r}} \dd{\nu(t,r)} \\
	&= \int_{\bbR \times [0,1]} \int_{\bbR \times [0,1]} \ev{u_{n,t',r'}, u_{n,t,r}} \dd{\nu^*(t',r')} \dd{\nu(t,r)}.
\end{align*}
By the dominated convergence theorem, we have
\[
	\norm{v_n}^2 \to \int_{\bbR \times [0,1]} \int_{\bbR \times [0,1]} \ev{u_{t',r'}, u_{t,r}} \dd{\nu^*(t',r')} \dd{\nu(t,r)} = \norm{v}^2.
\]
Similarly, we have
\[
	\ev{v, v_n} = \int_{\bbR \times [0,1]} \ev{v, u_{n,t,r}} \dd{\nu(t,r)}  \to \int_{\bbR \times [0,1]} \ev{v, u_{t,r}} \dd{\nu(t,r)} = \norm{v}^2.
\]
It follows that $\norm{v_n - v}^2 \to 0$, as desired.
\end{proof}

It is interesting to note that in the case where $f(x) = e^{itx}$, the Daletskii--Krein formula reduces to the \textit{Duhamel formula}:
\[
	\frac{d}{ds} e^{itA(s)} = \int_0^1 it  e^{itrA(s)} \Phi'(s) e^{it(1-r)A(s)} \dd{r}.
\]

\section{An Application from Quantum Statistical Mechanics}\label{sec:application}

We now show how the theory of double operator integrals can be used to prove a theorem that is important in quantum statistical mechanics. The main results of this section were originally proven in \cite[\S2]{BMNS_automorphic_equivalence} and were expanded upon in \cite[\S2 \& \S6]{NSY_quasilocality1}. We will follow their arguments quite closely but highlight where and how the theory of double operator integrals is used.

Before getting into details, we state a few elementary results that will be used in the proof of the main theorems. The first result below is essentially \cite[Thm.~III.6.17]{Kato}, in the special case of a self-adjoint operator on a Hilbert space. A simpler proof than the one in \cite{Kato} may be given in this special case, so we have provided it.  Note that throughout this section, we let $\sH$ be a Hilbert space.

\begin{theorem}[{cf.~\cite[Thm.~III.6.17]{Kato}}]\label{thm:Kato_integral}
Let $H$ be a densely defined self-adjoint operator on $\sH$. Let $\Gamma$ be a positively oriented piecewise $C^1$ simple closed curve in $\bbC$, disjoint from $\sigma(H)$, and let $P$ be the spectral projection corresponding to the portion of $\sigma(H)$ contained in the bounded component of the complement of the image of $\Gamma$. With $R(z) = (H - z)^{-1}$ for $z \in \bbC \setminus \sigma(H)$ denoting the resolvent of $H$, the projection $P$ is given by
\[
	P = -\frac{1}{2\pi i} \int_\Gamma R(z) \dd{z},
\]
where the integral is a Bochner integral.
\end{theorem}

\begin{proof}
Note that the integral exists since $R(z)$ is a norm-continuous function of $z \in \bbC \setminus \sigma(H)$. Indeed, if $z, z_0 \in \bbC \setminus \sigma(H)$, then 
\[
	R(z) - R(z_0) = (z - z_0)R(z) R(z_0).
\]
For any $z \in \bbC \setminus \sigma(H)$, we have
\[
	\norm{R(z)} \leq \frac{1}{\min\limits_{\lambda \in \sigma(H)} \abs{\lambda - z}}.
\]
Since the right hand side is bounded above in a neighborhood of $z_0$, we see that $\norm{R(z) - R(z_0)} \to 0$ as $z \to z_0$.

Let $E$ be the spectral measure of $H$ and let $\sigma_1$ be the portion of $\sigma(H)$ contained in the bounded component of $\bbC \setminus \Gamma$. 
For any $v,w \in \sH$ we have
\[\langle v, P w \rangle = \int_{\sigma(H)} \chi_{\sigma_1} \dd{E_{v,w}}\]
and
\begin{align*}
\left\langle v, \left(-\frac{1}{2\pi i} \int_\Gamma R(z) \dd{z}\right) w \right\rangle &= -\frac{1}{2\pi i} \int_\Gamma \langle v, R(z) w \rangle \dd{z}\\
&= -\frac{1}{2\pi i} \int_\Gamma \int_{\sigma(H)} \frac{1}{x - z} \dd{E_{v,w}(x)} \dd{z} \\ 
&= \int_{\sigma(H)} \left( \frac{1}{2\pi i}\int_\Gamma \frac{1}{z - x} \dd{z} \right) \dd{E_{v,w}(x)}.
\end{align*}
So it suffices to observe
\[\chi_{\sigma_1}(x) = \frac{1}{2\pi i}\int_\Gamma \frac{1}{z - x} \dd{z}\]
by elementary complex analysis.
\end{proof}

For our next preliminary results, let us prove the product and quotient rules for strongly differentiable functions. Recall that the notion of strong differentiability was defined and discussed in Definition \ref{defn:strongly_diff}. 

\begin{prop}
Let $J \subset \bbR$ be an interval. Given strongly differentiable maps $A:J \to \fB(\sH)$ and $B:J \to \fB(\sH)$, the map $C:J \to  \fB(\sH)$ given by $C(s) = A(s)B(s)$ is strongly differentiable and 
\[
C'(s) = A'(s)B(s) + A(s)B'(s)
\]
for all $s \in J$.
\end{prop}

\begin{proof}
Let $v \in \sH$ and fix $s \in J$. For any $t \in J \setminus \{s\}$ we have
\begin{equation}
\frac{C(t)v - C(s) v}{t-s} = \frac{A(t)B(t)v - A(t)B(s) v}{t-s} + \frac{A(t)B(s)v - A(s)B(s) v}{t-s}\,.
\end{equation}
The second term converges to $A'(s) B(s) v$ as $t \to s$ by strong differentiability of $A$. The first term can be rewritten as
\[\qty[\frac{A(t)B(t)v - A(t)B(s) v}{t-s} - A(t) B'(s) v ] + A(t) B'(s) v\]
As $t \to s$, the third term converges to $A(s) B'(s) v$. The norm of the term in square brackets is bounded above by
\[\norm{A(t)} \norm{\frac{B(t)v - B(s) v}{t-s} - B'(s) v }.\]
This converges to $0$ as $t \to s$ by strong differentiability of $B(s)$ and since $\norm{A(t)}$ is bounded for $t$ in a neighborhood of $s$. Thus, we have from our original equation
 \[\lim_{t\to s}\frac{C(t)v - C(s) v}{t-s} = A'(s) B(s) v + A(s) B'(s) v\,.\]
The statement follows since $v$ and $s$ were chosen arbitrarily. 
\end{proof}

\begin{prop}
Let $J \subset \bbR$ be an interval. Suppose $A:J \to \fB(\sH)$ is strongly differentiable and  $A(s)$ is invertible for each $s \in J$. Then the function $A^{-1}:J \to  \fB(\sH)$ given by $A^{-1}(s) = A(s)^{-1}$ is strongly differentiable and 
\[
	\left(A^{-1}\right)'(s) = -A(s)^{-1} A'(s) A(s)^{-1}
\]
for each $s \in J$.
\end{prop}

\begin{proof}
Let $v \in \sH$ and fix $s \in J$. For any $t \in J \setminus \qty{s}$, we have
\begin{align*}
\frac{A^{-1}(t)v - A^{-1}(s)v}{t - s} &= -\frac{A(t)^{-1}[A(t) - A(s)]A(s)^{-1}v}{t - s} \\
&= -\frac{[A(t)^{-1} - A(s)^{-1}][A(t) - A(s)] A(s)^{-1}v}{t - s}\\
&\qquad - \frac{A(s)^{-1}[A(t) - A(s)]A(s)^{-1}v}{t - s}.
\end{align*}
The last term converges to $-A(s)^{-1}A'(s)A(s)^{-1}v$ as $t \to s$ by strong differentiability of $A$. The first term has norm bounded above by
\[
	\norm{A(t)^{-1} - A(s)^{-1}}\norm{\frac{A(t) - A(s)}{t - s} A(s)^{-1}v}.
\]
As $t \to s$, the first norm above converges to $0$ by norm-continuity of $A$, while the second term converges to $A'(s)A(s)^{-1}v$. Therefore the product converges to $0$, concluding the proof.
\end{proof}

Let us now explain the setup of the main theorems. 

\begin{definition}
Let $J \subset \bbR$ be an interval.
We say $\Phi:J \to \fB(\sH)$ is \emph{strongly $C^1$} if $\Phi$ is strongly differentiable and the derivative map $\Phi':J \to \fB(\sH)$ is strongly continuous. 

Note that if a map $\Phi:J \to \fB(\sH)$ is strongly continuous, then the function $J \to \bbR$, $s \mapsto \norm{\Phi(s)}$ is lower semicontinuous as it is the supremum of the continuous functions $s \mapsto \norm{\Phi(s)v}$ over all $v \in \sH$ with $\norm{v} \leq 1$. In addition, if $K$ is a compact subset of $J$, then $\norm{\Phi(s)}$ is bounded on $K$ as a consequence of the uniform boundedness principle. Replacing $\Phi$ with $\Phi'$, we see that if $\Phi$ is strongly $C^1$, then $\norm{\Phi'(s)}$ is lower semicontinuous and bounded on compact subsets of $J$. 
\end{definition}

\begin{assumption}\label{assmpt1}
Let $J \subset \bbR$ be an interval. Let $\Phi:J \to \fB(\sH)$ be strongly $C^1$ and pointwise self-adjoint. We let $H_0$ be a densely defined and self-adjoint operator on $\sH$ and we define
\[
H(s) = H_0 + \Phi(s)
\]
for all $s \in J$. We denote the spectrum of $H(s)$ by $\sigma(s)$. We assume there exist compact intervals $I(s)$ defined for all $s \in J$, with endpoints depending continuously on $s$, such that:
\begin{enumerate}
	\item $\sigma_1(s) \defeq I(s) \cap \sigma(s)$ and $\sigma_2(s) \defeq \sigma(s) \setminus \sigma_1(s)$ are nonempty for all $s$,
	\item there exists $\gamma > 0$ such that 
	\[
		\inf\qty{\abs{x_1 - x_2} : x_1 \in I(s), x_2 \in \sigma_2(s)} \geq \gamma
	\]
	for all $s$.
\end{enumerate}
Finally, to establish notation, we let $E_s$ denote the spectral measure of $H(s)$ and we let $P(s)$ denote the spectral projection onto the first isolated component of the spectrum $\sigma_1(s)$, i.e.,
\[
	P(s) = \int_{\sigma_1(s)} \dd{E_s}.
\]
Given $z \in \bbC \setminus \sigma(s)$, we let
\[
	R(z,s) = (H(s) - z)^{-1}
\]
be the resolvent of $H(s)$ at $z \in \bbC \setminus \sigma(s)$. We let $\Gamma(s) : [0,2\pi] \to \bbC$ be the circle
\begin{equation}
	\Gamma(s)(\theta) = c(s) + \frac{1}{2}d(s) e^{i\theta} 
\end{equation}
where
\[
	c(s) = \frac{\min I(s) + \max I(s)}{2} \qqtext{and} d(s) = \max I(s) - \min I(s) + \gamma.
\]
\end{assumption}

Our goal will be to prove strong differentiability of $P$ and moreover that its derivative is given by a formula
\begin{equation}\label{eq:proj_evolution}
P'(s) = i[D(s), P(s)]\,,
\end{equation}
where $D:J \to \fB(\sH)$ is a family of self-adjoint operators depending explicitly on $H(s)$, to be defined later. Each side will be exhibited as a linear combination of double operator integrals, at which point the theory developed in the previous sections can be used to manipulate one into the other.  Before we define $D$, let us focus on the left hand side of \eqref{eq:proj_evolution}. 

\begin{lemma}\label{lem:P'(s)}
Under Assumption \ref{assmpt1}, $P(s)$ is strongly differentiable, with derivative 
\begin{equation}\label{eq:First P'(s)}
P'(s) = \frac{1}{2\pi i} \int_{\Gamma(s)} R(z,s) \Phi'(s) R(z,s) \dd{z}, 
\end{equation}
where the integral is a Bochner integral. 
\end{lemma}

\begin{proof}
Fix $s_0 \in J$ and let $\Gamma = \Gamma(s_0)$ for ease of notation. We first note that since the endpoints of $I(s)$ depend continuously on $s$, there exists a connected set $U$ containing $s_0$ that is relatively open in $J$ and such that for all $s \in U$, we have that $I(s)$ is contained in the bounded component of $\bbC \setminus \Gamma([0,2\pi])$, $\sigma_2(s)$ is contained in the unbounded component of $\bbC \setminus \Gamma([0,2\pi])$, and the distance from $\sigma(s)$ to $\Gamma([0,2\pi])$ is greater than $\gamma/3$. It follows from Theorem \ref{thm:Kato_integral} that
\[
	P(s) = -\frac{1}{2\pi i} \int_\Gamma R(z,s) \dd{z}
\]
for all $s \in U$, i.e., the same contour $\Gamma$ may be used for all $s \in U$. 

Next, we note that $[0,2\pi] \times U \to \fB(\sH)$, $(\theta, s) \mapsto R(\Gamma(\theta), s)$ is jointly norm-continuous in $\theta$ and $s$. Indeed, given pairs $(\theta_1, s_1)$ and $(\theta_2, s_2)$ in $[0,2\pi] \times U$, if we denote $R_i = R(\Gamma(\theta_i), s_i)$ for $i \in \qty{1,2}$, then we have the equation
\[
	R_2 - R_1 = R_2\qty[(\Gamma(\theta_2) - \Gamma(\theta_1)) - (\Phi(s_2) - \Phi(s_1))]R_1.
\]
Thus, as $(\theta_2, s_2) \to (\theta_1, s_1)$ we have $R_2 - R_1 \to 0$ since $\Gamma$ is continuous, $\Phi$ is norm-continuous, and $\norm{R(\Gamma(\theta), s)} \leq 3/\gamma$ for all $\theta$ and $s$. In particular, since $z \mapsto R(z,s_0)\Phi'(s_0)R(z,s_0)$ is norm-continuous, the integral in \eqref{eq:First P'(s)} (with $s = s_0$) exists as a Bochner integral.

Moreover, for fixed $\theta$, the above equation and the strong differentiability of $\Phi$ imply that $R(\Gamma(\theta), s)$ is strongly differentiable in $s$ with derivative
\[
	\frac{\partial}{\partial s} R(\Gamma(\theta), s) = -R(\Gamma(\theta), s)\Phi'(s)R(\Gamma(\theta),s).
\]
Since $R(\Gamma(\theta),s)$ is jointly norm-continuous and $\Phi'(s)$ is strongly continuous (hence locally norm-bounded), we see that the above partial derivative is jointly strongly continuous as a function $[0,2\pi] \times U \to \fB(\sH)$. 

Since for any $v \in \sH$, we know that $R(\Gamma(\theta), s)v$ is continuous and has continuous partial derivative with respect to $s$, we may differentiate under the integral sign \cite[Thm.~XIII.8.1]{LangRealAndFunctionalAnalysis} to see that $P(s)v$ is differentiable with derivative at $s_0$ given by
\begin{align*}
P'(s_0)v &= -\frac{1}{2\pi i} \int_\Gamma \frac{\partial}{\partial s} R(z,s)v \bigg|_{s = s_0}  \dd{z} \\
&= \frac{1}{2\pi i } \int_\Gamma R(z, s_0)\Phi'(s_0)R(z,s_0)v \dd{z}
\end{align*}
as desired.
\end{proof}

Let us now consider the right hand side of \eqref{eq:proj_evolution}. The operator $D(s)$ is defined in terms of a ``weight function'' $w_\gamma$, where $\gamma$ is as in Assumption \ref{assmpt1}. 

\begin{assumption}\label{assmpt2}
Let $w_\gamma \in L^1(\rr)$ be a real-valued function such that $\int_\rr w_\gamma(t) dt = 1$, $\int_\bbR \abs{t w_\gamma(t)} \dd{t} < \infty$ and $\supp\, \hat w_\gamma \subset [-\gamma,\gamma]$, where $\hat w_\gamma$ is the Fourier transform of $w_\gamma$.
\end{assumption}

Recall that our convention for the Fourier transform is:
\[
	\hat w_\gamma(\xi) \defeq \int_\bbR w_\gamma(t) e^{-i\xi t} \dd{t}.
\]

For example, such a function $w_\gamma$ can be obtained by taking a smooth function $f:\bbR \to \bbR$ compactly supported in $[-\gamma, \gamma]$ satisfying $f(k) = f(-k)$ and $f(0) = 1$, and defining
\[
	w_\gamma = \Re \cF^{-1}f,
\]
where $\cF^{-1}f$ is the inverse Fourier transform of $f$. For further applications in the context of quantum spin systems it is important to have a more concrete description of $w_\gamma$ with more control over its decay as $\abs{t} \to \infty$, but for our purposes the above description is sufficient. A more explicit construction of a function $w_\gamma$ with explicit bounds on the decay as $\abs{t} \to \infty$ can be found in \cite[\S VI.B]{NSY_quasilocality1}.

Under Assumption \ref{assmpt1} and with $w_\gamma$ as in Assumption \ref{assmpt2}, we define the \textit{Hastings generator} for $s \in J$ as the map $D(s):\sH \to \sH$ defined by
\begin{equation}\label{eq:hastings}
D(s)v = \int_{-\infty}^\infty w_\gamma(t) \int_0^t e^{iuH(s)} \Phi'(s) e^{-iuH(s)}v \dd{u} \dd{t}.
\end{equation}
Note that $u \mapsto e^{iuH(s)}\Phi'(s)e^{-iuH(s)}v$ is continuous and therefore the integral with respect to $u$ is a well-defined Bochner integral with 
\[
	\norm{\int_0^t e^{iuH(s)}\Phi'(s)e^{-iuH(s)}v \dd{u}} \leq \abs{t} \norm{\Phi'(s)} \norm{v}.
\]
It is also clear that the above integral is a continuous function of $t$ and therefore the integral with respect to $t$ in \eqref{eq:hastings} is a well-defined Bochner integral as well. We see that $D(s)$ is a bounded linear operator with 
\[
	\norm{D(s)} \leq \norm{\Phi'(s)} \int_{-\infty}^\infty \abs{t w_\gamma(t)} \dd{t}.
\]
It is also easy to check that $D(s)$ is self-adjoint.

\begin{theorem}\label{thm:Hastings_commutator}
Under Assumption \ref{assmpt1}, with $w_\gamma$ as in Assumption \ref{assmpt2}, and with $D(s)$ as in \eqref{eq:hastings}, we have
\begin{equation}\label{eq:Hastings_commutator_again}
	P'(s) = i[D(s), P(s)]
\end{equation}
for all $s \in J$.
\end{theorem}

\begin{proof}
The value of $s$ will remain fixed throughout the proof, so we set 
\[
	\Gamma = \Gamma(s), \quad \sigma = \sigma(s), \quad \sigma_1 = \sigma_1(s), \quad \tn{and} \quad \sigma_2 = \sigma_2(s)
\] 
for ease of notation. Define the maps $\alpha, \beta: \sigma \times [0,2\pi] \to \bbC$ by
\[\alpha(x, \theta) = \frac{\chi_{\sigma_1}(x)}{x - \Gamma(\theta)} \qqtext{and}\beta(x, \theta) = \frac{\chi_{\sigma_2}(x)}{x - \Gamma(\theta)}\,.\]
Since $\Gamma$ never intersects $\sigma$, the maps $\alpha$ and $\beta$ are continuous functions on $\sigma \times [0,2\pi]$. In particular, they are bounded and measurable. Then we can define $\phi_1, \phi_2: \sigma \times \sigma \to \bbC$ by strict decompositions
\[\phi_1(x, y) = \int_{\Gamma} \alpha(x,z) \beta(y, z) \dd{z} \qqtext{and} \phi_2(x, y) = \int_{\Gamma} \beta(x, z) \alpha(y,z) \dd{z},\]
where the auxiliary measure space is $[0,2\pi]$ with complex measure $\Gamma'\! \dd{\theta}$ (see Remark \ref{rem:complex_measure}). Thus $\phi_1, \phi_2 \in \fM_{00}$. Note, by an application of the Cauchy integral formula we have
\[\phi_1(x, y) = \begin{cases} 0 &\tn{if } x \notin \sigma_1 \text{ or } y \notin \sigma_2 \\ \frac{2\pi i}{x - y} & \text{otherwise}\end{cases}\]
and
\[\phi_2(x, y) = \begin{cases} 0 &\tn{if } x \notin \sigma_2 \text{ or } y \notin \sigma_1 \\ \frac{2\pi i}{y - x} & \text{otherwise.}\end{cases}\]

Since $P(s) = P(s)^2$, the product rule yields $P'(s) = P(s)P'(s) + P'(s)P(s)$. Multiplying by $P(s)$ on the left and right gives $P(s)P'(s)P(s) = 0$. Combining these two equations gives
\begin{equation}\label{eq:Second P'(s)}
P'(s) = P(s)P'(s)(1-P(s)) + (1-P(s))P'(s)P(s).
\end{equation}
Let $E$ be the spectral measure associated to $H(s)$. By Lemma \ref{lem:P'(s)}  and \eqref{eq:Second P'(s)} we have
\begin{align*}
2\pi i P'(s) &=  \int_{\Gamma} P(s)R(z,s) \Phi'(s) R(z,s) (1-P(s)) \dd{z} \nonumber \\
&\quad \quad +  \int_{\Gamma} (1-P(s)) R(z,s) \Phi'(s) R(z,s) P(s) \dd{z} \nonumber \\
&= \int_{\Gamma} \left[ \left(\int_{\sigma}\alpha(x,z) \dd{E(x)} \right) \Phi'(s)\left(\int_{\sigma} \beta(y,z) \dd{E(y)} \right) \right] \dd{z} \nonumber \\
&\quad \quad +  \int_{\Gamma} \left[\left(\int_{\sigma}\beta(x,z) \dd{E(x)} \right)  \Phi'(s) \left(\int_{\sigma}\alpha(y,z) \dd{E(y)} \right)  \right] \dd{z}. \nonumber
\end{align*}
Given our strict decompositions of $\phi_1$ and $\phi_2$ we can write $P'(s)$ in terms of double operator integrals as
\begin{equation} \label{eq:DOI_phi}
P'(s) = \frac{1}{2\pi i} \left(\int_{\sigma \times \sigma} \phi_1 \dd{\cG} \Phi'(s) +  \int_{\sigma \times \sigma} \phi_2 \dd{\cG} \Phi'(s)\right),
\end{equation}
where $\cG$ is the projection valued measure on $\sigma \times \sigma$ constructed from two copies of $E$.

We now rewrite the right hand side of \eqref{eq:Hastings_commutator_again} as a linear combination of double operator integrals. We again begin by defining two functions through strict decompositions. Let
\[
	Q = \qty{(t,u) \in \bbR \times \bbR: \abs{u} \leq \abs{t} \tn{ and } tu \geq 0}.
\]
Note that 
\[
	\int_\bbR \int_\bbR \abs{\chi_{Q}(t,u)w_\gamma(t)} \dd{u} \dd{t} = \int_\bbR \abs{tw_\gamma(t)} \dd{t} < \infty.
\]
Thus, we may define a complex measure by $\dd{\nu(t,u)} = \chi_Q(t,u)w_\gamma(t) \dd{t}\dd{u}$. Now define the maps  $\alpha_j, \beta_j: \sigma \times \bbR \times \bbR \to \bbC$ for $j=1,2$ by
\[\alpha_j(x, t, u) = \chi_{\sigma_2}(x) e^{(-1)^{j+1} iux} \qqtext{and} \beta_j(x, t, u) = \chi_{\sigma_1}(x) e^{(-1)^j iux}\,.\]
Note there is intentionally no dependence on $t$. These are clearly bounded measurable functions. Then we can define functions $\psi_1, \psi_2:\sigma \times \sigma  \to \bbC$ by strict decompositions
\[
	\psi_1(x,y) = \int_{\bbR \times \bbR} \alpha_1(x,t,u)\beta_1(y,t,u) \dd{\nu(t,u)}
\]
and
\[
	\psi_2(x,y) = \int_{\bbR \times \bbR} \beta_2(x,t,u)\alpha_2(y,t,u) \dd{\nu(t,u)}.
\]

Now, observe that 
\[
(1-P(s)) D(s) P(s) - P(s) D(s) (1-P(s)) = [D(s), P(s)]
\]
so that
\begin{align}
[D(s), P(s)]  &= \int_\bbR w_\gamma(t) \int_0^t (1-P(s)) e^{iuH(s)} \Phi'(s) e^{-iuH(s)} P(s) \dd{u} \dd{t} \nonumber \\
&\quad  - \int_\bbR w_\gamma(t) \int_0^t P(s) e^{iuH(s)} \Phi'(s) e^{-iuH(s)} (1-P(s)) \dd{u} \dd{t}. \nonumber 
\end{align}
For any $v,w \in \sH$, acting the first line of the right hand side on $w$ and taking an inner product with $v$ yields
\[
	\int_{\bbR \times \bbR} \ev{v, \int_\sigma \alpha_1(x,t,u)\dd{E(x)}\Phi'(s) \int_\sigma \beta_1(y,t,u) \dd{E(y)} w} \dd{\nu(t,u)}.
\]
Given our strict decomposition of $\psi_1$, we can write this as 
\[
\ev{v, \int_{\sigma \times \sigma} \psi_1 \dd{\cG}\qty(\Phi'(s)) w}.
\] Similar considerations apply to the second term on the right hand side of the equation for $[D(s),P(s)]$, so we get
\begin{equation}\label{eq:doi_[D(s),P(s)]}
	i[D(s),P(s)] = \int_{\sigma \times \sigma} i\psi_1 \dd{\cG} \Phi'(s) - \int_{\sigma \times \sigma} i\psi_2 \dd{\cG} \Phi'(s).
\end{equation}

Comparing \eqref{eq:DOI_phi} and \eqref{eq:doi_[D(s),P(s)]}, we see that the theorem will be proven if we can show 
\[
	\frac{1}{2\pi i} \phi_1 = -i\psi_2 \qqtext{and} \frac{1}{2\pi i} \phi_2 = i\psi_1.
\]
Thus, we compute $\psi_1(x,y)$:
\begin{align*}
i\psi_1(x,y) = i\int_\bbR w_\gamma(t) \int_0^t \chi_{\sigma_2}(x)\chi_{\sigma_1}(y) e^{iu(x-y)} \dd{u} \dd{t}.
\end{align*}
We note that $\psi_1(x,y) = \phi_2(x,y) =0$ if $x \notin \sigma_2$ or $y \notin \sigma_1$. If $x \in \sigma_2$ and $y \in \sigma_1$, then we have
\begin{align*}
i\psi_1(x,y) &= \int_\bbR w_\gamma(t) \qty(\frac{e^{it(x-y)} - 1}{x-y}) \dd{t} = \frac{\hat w_\gamma(y-x) - 1 }{x-y}.
\end{align*}
Since $\supp\,\hat w_\gamma \subset [-\gamma, \gamma]$ and $\abs{y - x} \geq \gamma$, we have $\hat w_\gamma(y-x) = 0$. Thus,
\[
	i\psi_1(x,y) = \frac{1}{y-x} = \frac{1}{2\pi i} \phi_2(x,y),
\]
as desired. The calculation that $\phi_1/2\pi i = -i\psi_2$ is similar. This concludes the proof.
\end{proof}

The operator $D(s)$ is called a generator because it can be used as an $s$-dependent Hamiltonian generating a family of unitaries. More precisely, it is shown in \cite[Prop.~2.2]{NSY_quasilocality1} that if $A:J \to \fB(\sH)$ is strongly continuous and pointwise self-adjoint and $s_0 \in J$, then there exists a unique strongly differentiable map $U:J \to \fB(\sH)$ such that $U(s_0) = \1$ and
\[
	U'(s) = iA(s)U(s)
\]
for all $s \in J$. Furthermore, $U(s)$ is unitary for all $s \in J$.

In order to apply this result with the Hastings generator in place of $A(s)$, we need to know that $D:J \to \fB(\sH)$ is strongly continuous. Recall from Lemma \ref{lem:exponential_difference} that
\[
	\norm{e^{iuH(s_2)} - e^{iuH(s_1)}} \leq \abs{u}\norm{\Phi(s_2) - \Phi(s_1)}
\]
for all $s_2, s_1 \in J$. In particular, norm-continuity of $s \mapsto e^{iuH(s)}$ follows from norm-continuity of $\Phi$.  Then for any $r \in J$ and $v \in \sH$ we have
\begin{align*}
\lim_{s \to r} D(s) v 
&= \lim_{s \to r} \int_{-\infty}^\infty w_\gamma(t) \int_0^t e^{iuH(s)} \Phi'(s) e^{-iuH(s)}v \dd{u} \dd{t} \\
&=  \int_{-\infty}^\infty w_\gamma(t) \int_0^t e^{iuH(r)} \Phi'(r) e^{-iuH(r)} v \dd{u} \dd{t} = D(r)v,
\end{align*}
where the limit is passed through the integrals by the dominated convergence theorem, using the fact that $\int_\bbR \abs{t w_\gamma(t)} \dd{t} < \infty$ and the fact that $\|\Phi'(s)\|$ is uniformly bounded on compact sets.

Thus, \cite[Prop.~2.2]{NSY_quasilocality1} implies that there exists a unique strongly differentiable map $U:J \to \fB(\sH)$ such that
\begin{equation}\label{eq:Hastings_U}
	U(s_0) = \1 \qqtext{and} U'(s) = iD(s)U(s).
\end{equation}
Note that the quotient rule for strong differentiability implies that $U(s)^*$ is also strongly differentiable with strong derivative
\[
	\frac{d}{ds} U(s)^* = -U(s)^*U'(s)U(s)^* = -iU(s)^*D(s).
\]

\begin{thm}
Let Assumption \ref{assmpt1} be given, let $w_\gamma$ be as in Assumption \ref{assmpt2}, let the Hastings generator $D(s)$ be as in \eqref{eq:hastings}, let $s_0 \in J$, and let $U:J \to \fB(\sH)$ be the unique strongly differentiable map satisfying \eqref{eq:Hastings_U}. Then the spectral projections $P(s)$ associated with the isolated portion of the spectrum $\sigma_1(s)$ are given by
\[P(s) = U(s)P(s_0)U(s)^*\]
for all $s \in J$.
\end{thm}

\begin{proof}
Note $U(s_0)^* P(s_0) U(s_0) = P(s_0)$ since $U(s_0)=\mathds{1}$ by our initial condition. Furthermore, $U(s)^*P(s)U(s)$ is strongly differentiable by the product rule, with derivative:
\begin{align*}
\frac{d}{ds}U(s)^*P(s)U(s) &= -iU(s)^*D(s)P(s)U(s) + iU(s)^*[D(s),P(s)]U(s) \\
&\qquad \qquad + iU(s)^*P(s)D(s)U(s)\\
&= 0.
\end{align*}
It follows that $U(s)^*P(s)U(s)$ is constant at $P(s_0)$. The result follows by conjugating with $U(s)$.
\end{proof}

\section*{Data Availability Statement}

Data sharing is not applicable to this article as no new data were created or analyzed in this study.

\section*{Author Declarations}

The authors have no conflicts to disclose.

\bibliographystyle{amsalpha}
\bibliography{dbl_op_int.bib}

\end{document}